\documentclass[conference]{IEEEtran}
\usepackage{amsmath}
\usepackage{amsthm}
\usepackage{amssymb}
\usepackage{epsfig}
\usepackage{subfigure}
\usepackage{color}
\usepackage{cite}
\newtheorem{lem}{Lemma}
\IEEEoverridecommandlockouts
\begin{document}
\title{A Non-stationary Service Curve Model for Performance Analysis of Transient Phases}
\author{\IEEEauthorblockN{Nico Becker and Markus Fidler}
\thanks{This work was supported by an ERC Starting Grant (UnIQue, StG 306644).}
\IEEEauthorblockA{Institute of Communications Technology, Leibniz Universit\"{a}t Hannover}}
\maketitle
\begin{abstract}
Steady-state solutions for a variety of relevant queueing systems are known today, e.g., from queueing theory, effective bandwidths, and network calculus. The behavior during transient phases, on the other hand, is understood to a much lesser extent as its analysis poses significant challenges. Considering the majority of short-lived flows, transient effects that have diverse causes, such as TCP slow start, sleep scheduling in wireless networks, or signalling in cellular networks, are, however, predominant. This paper contributes a general model of regenerative service processes to characterize the transient behavior of systems. The model leads to a notion of non-stationary service curves that can be conveniently integrated into the framework of the stochastic network calculus. We derive respective models of sleep scheduling and show the significant impact of transient phases on backlogs and delays. We also consider measurement methods that estimate the service of an unknown system from observations of selected probe traffic. We find that the prevailing rate scanning method does not recover the service during transient phases well. This limitation is fundamental as it is explained by the non-convexity of non-stationary service curves. A second key difficulty is proven to be due to the super-additivity of network service processes. We devise a novel two-phase probing technique that first determines a minimal pattern of probe traffic. This probe is used to obtain an accurate estimate of the unknown transient service.
\end{abstract}
%
%
\section{Introduction}
\label{sec:introduction}
The majority of flows in today's computer networks are short-lived~\cite{mellia:shortlivedtcp} and hence dominated by various transient effects that can have a significant impact on their performance. Relevant examples include TCP slow start~\cite{mellia:shortlivedtcp}, where the size of the initial congestion window has been repeatedly under debate~\cite{dukkipati:tcpinitialwindow}, the convergence of routing protocols, power saving in wireless networks using polling, or signalling and discontinuous reception in cellular networks~\cite{becker:lte}.

The performance analysis of transient phases causes, however, fundamental difficulties and queueing theory offers solutions mainly for the steady-state. As an example consider the basic M$\mid$M$\mid$1 queue, where the stationary state distribution follows readily from a set of linear balance equations, e.g.,~\cite{ross:probabilitymodels}. The transient behavior, on the other hand, is expressed by a set of differential equations for which mainly approximate or numerical solutions are known~\cite{zhang:transientmm1}. As a consequence, analytical works on transient effects in computer networks are sparse~\cite{wang:transientatm,souza1998algorithm,horvath2012transient} and often tailored to specific problems like TCP congestion control~\cite{mellia:shortlivedtcp}.

A theory that does without an assumption of stationarity is the deterministic network calculus~\cite{leboudec:networkcalculus, chang:performanceguarantees}. It employs envelope functions of possibly non-stationary processes to analyze the worst-case behavior of systems. Hence, it takes transient phases into the consideration. Once the worst-case is achieved, it can, however, not predict how a system progresses. We explain this effect and present an example in Sec.~\ref{sec:systemmodelregenerative}. The stochastic network calculus\cite{chang:performanceguarantees, burchard:endtoendstatisticalcalculus, li:effectivebandwidthcalculus2, ciucu:networkservicecurvescaling2, fidler:momentcalculus, jiang:stochasticnetworkcalculus, fidler:netcalcsurvey, ciucu:goodvalue, fidler:netcalcguide}, on the other hand, typically either assumes stationarity or uses stationary bounds.

In this work, we use the notion of time-variant systems~\cite{chang:dynamicserviceguarantees, agrawal:timevaryingservice} to model non-stationary service characteristics. Time-variant systems are described by bivariate instead of univariate functions to consider changes over time. We show that non-stationarity can be dealt with in the stochastic network calculus using a similar notational extension. We contribute a non-stationary service curve model that characterizes the service of systems during transient phases. While many known results are recovered for the extended model, certain fundamental properties such as commutativity~\cite{chang:dynamicserviceguarantees} differ. We derive solutions for systems with sleep scheduling, provide insights into the transient behavior, and quantify the considerable transient overshoot.
Secondly, we examine methods for estimation of a system's service curve from measurements of probe traffic. We refine known measurement methods for estimation of non-stationary service curves, where we encounter additional difficulties that are attributed to the non-convexity and super-additivity of the service. We devise a novel minimal probing method that estimates a non-stationary service curve and provides a measure of accuracy.

The remainder of this work is structured as follows. In Sec.~\ref{sec:systemmodel}, we define non-stationary service curves, show a method for construction, and derive models of systems with sleep scheduling. In Sec.~\ref{sec:estimation}, we investigate the measurement-based estimation of non-stationary service curves. We reveal difficulties that arise and devise a new minimal probing method. We discuss further related works in the respective sections. Sec.~\ref{sec:conclusion} presents brief conclusions. While we restrict the exposition to non-stationary systems, we note that non-stationary traffic can be dealt with in the same way.
%
%
\section{System Model}
\label{sec:systemmodel}
In this section, we derive a model of non-stationary service curves (Sec.~\ref{sec:systemmodelrandom}) that enable analyzing the performance of systems during transient phases. The basis of this model is a definition of regenerative service processes (Sec.~\ref{sec:systemmodelregenerative}), where regeneration points mark the start of new transient phases. We include solutions for sleep scheduling and show its impact on the performance.
%
%
\subsection{Regenerative Service Processes}
\label{sec:systemmodelregenerative}
We consider a system with cumulative arrivals $A(t)$, where $A(t)$ denotes the number of bits that arrive in the time interval $(0,t]$. By convention, there are no arrivals for $t \le 0$ so that we generally consider $t \ge 0$. Clearly, $A(t)$ is a non-negative, non-decreasing function, and $A(0)=0$. Shorthand notation $A(\tau,t) = A(t) - A(\tau)$ is used to denote the arrivals in $(\tau,t]$ where $t \ge \tau \ge 0$. Trivially, $A(t,t)=0$ for all $t \ge 0$. Similarly, $D(t)$ denotes the cumulative departures from the system.

The service that is provided by the system is characterized by a time-variant service process $S(\tau,t)$ that establishes the departure guarantee~\cite{chang:performanceguarantees,chang:timevaryingfiltering,chang:dynamicserviceguarantees,agrawal:timevaryingservice}
\begin{equation}
D(t) \ge \inf_{\tau \in [0,t]} \{ A(\tau) + S(\tau,t) \} =: A \otimes S(t).
\label{eq:serviceprocess}
\end{equation}
By convention, $S(\tau,t)$ is non-negative and $S(t,t) = 0$ for all $t \ge 0$. The operator $\otimes$ that is defined by~\eqref{eq:serviceprocess} is known as convolution under a min-plus algebra~\cite{chang:performanceguarantees,leboudec:networkcalculus}. We note that $\otimes$ is associative but not commutative in general. Examples of~\eqref{eq:serviceprocess} include a work-conserving server with a time-variant capacity~\cite{chang:performanceguarantees} where $S(\tau,t)$ denotes the service that is available in the interval $(\tau,t]$; scheduling with cross-traffic~\cite{fidler:momentcalculus}; and networks of systems where the network service process $S^{net}(\tau,t)$ is computed from the service processes $S^i(\tau,t)$ of the individual systems $i=1 \dots n$ by recursive insertion of~\eqref{eq:serviceprocess} as $S^{net}(\tau,t) = S^1 \otimes S^2 \otimes \dots \otimes S^n(\tau,t)$~\cite{chang:performanceguarantees}. The service guarantee~\eqref{eq:serviceprocess} enables the derivation of performance bounds. An upper bound of the backlog $B(t) = A(t) - D(t)$ follows by insertion of~\eqref{eq:serviceprocess} as
\begin{equation}
B(t) \le \sup_{\tau \in[0,t]} \{A(\tau,t) - S(\tau,t)\}.
\label{eq:backlogbound}
\end{equation}
Similarly, the delay defined as $W(t) = \inf\{w \ge 0: A(t) \le D(t+w) \}$ can be considered.

Throughout this work, we assume that the service is a regenerative process~\cite{ross:probabilitymodels} with regeneration points $\mathbb{P} = \{P_0,P_1,P_2,\dots\}$ where $P_0=0$ and $P_i < P_{i+1}$ for all $i \ge 0$. We divide $S(\tau,t)$ into segments
\begin{equation}
S_i(\tau,t) = S(\tau+P_i,t+P_i)
\label{eq:regenerativeservicsample}
\end{equation}
for all $0 \le \tau \le t \le P_{i+1} - P_i$ and $i \ge 0$, where $S_i(\tau,t)$ is the service process between the $i$th and the $(i+1)$th regeneration point. The defining characteristic of a regenerative process is that the $S_i(\tau,t)$ are statistical replicas, i.e.,
\begin{equation}
\mathsf{P}[S_i(\tau,t) \le x] = \mathsf{P}[S_j(\tau,t) \le x]
\label{eq:statisticalreplica}
\end{equation}
for all $i,j,x \ge 0$, and $0 \le \tau \le t \le \min \{P_{i+1}\!-\!P_i, P_{j+1}\!-\!P_j \}$. Owing to~\eqref{eq:statisticalreplica}, we omit the index $i$ in the sequel. Also, we will not explicitly mention the constraint $t \le P_{i+1}$ and assume that the next regeneration point $P_{i+1}$ is spaced sufficiently apart.
%
%
\vspace{5pt}
\subsubsection*{Deterministic Sleep Scheduling}
We present a first application to sleep scheduling, where we consider a transmitter and a receiver that if idle go to sleep state according to a defined protocol. Wake up is scheduled deterministically, $T$ units of time after entering the sleep state. The transmission rate in sleep state is zero and otherwise it is $R$. Clearly, each transition to sleep state is a regeneration point and the time-variant service process follows for $t \ge \tau \ge 0$ as
\begin{equation*}
S(\tau,t) =
\begin{cases}
0, & t \le T \\
R (t-T), & t > T, \tau \le T \\
R (t-\tau), & t > T, \tau > T
\end{cases}
\end{equation*}
so that
\begin{equation}
S(\tau,t) = [R(t-\max\{\tau,T\})]_+
\label{eq:nonstationarylatencyrate}
\end{equation}
where $[x]_+ = \max \{0,x\}$ is the non-negative part of $x$.

For numerical evaluation we use a discrete time equivalent of a stationary Poisson arrival process: the number of packet arrivals $N(t)$ in an interval of length $t$ is binomial with parameter $\alpha \in [0,1]$. The individual packet sizes $Y(i)$ with index $i = 1,2,\dots$ are independent and identically distributed (iid) geometric random variables with parameter $\beta \in (0,1]$. Parameter $\alpha$ has the interpretation of an average arrival rate and $1/\beta$ is the average size of packets. For the special case of a system with constant service rate $R=1$, $\alpha/\beta$ is the utilization and $\alpha < \beta$ is required for stability.

We characterize the process using an upper envelope function that is derived from its moment generating function (MGF). The MGF of a random variable $X$ is defined as $\mathsf{M}_X(\theta) = \mathsf{E}[e^{\theta X}]$ for any $\theta$. The respective MGFs of the above processes are $\mathsf{M}_N(\vartheta,t) = (\alpha e^{\vartheta} + 1-\alpha)^t$ and $\mathsf{M}_Y(\theta) = \beta e^{\theta}/(1-(1-\beta)e^{\theta})$ for $\theta \in [0,-\ln(1-\beta))$~\cite{ross:probability}. The cumulative arrival process $A(t)$ is the doubly stochastic process $A(t) = \sum_{i=1}^{N(t)} Y(i)$. It has MGF $\mathsf{M}_A(\theta,t) = \mathsf{M}_N(\ln\mathsf{M}_Y(\theta),t)$~\cite{ross:probability, fidler:netcalcguide} so that by insertion
\begin{equation}
\mathsf{M}_A(\theta,t) = \biggl(\frac{\alpha \beta e^{\theta}}{1-(1-\beta)e^{\theta}} + 1-\alpha \biggr)^t .
\label{eq:poissonmgf}
\end{equation}

Using Chernoff's theorem $\mathsf{P}[X \ge x] \le e^{-\theta x} \mathsf{M}_X(\theta)$ for $\theta \ge 0$ and established methods of the stochastic network calculus~\cite{li:effectivebandwidthcalculus2, ciucu:networkservicecurvescaling2} it can be shown\footnote{We omit the proof as it is dual to the derivation of~\eqref{eq:statisticalserviceenvelope} and considers only the special case of a stationary process.} that the function
\begin{equation}
\mathcal{A}^{\varepsilon}(t) = \frac{1}{\theta(t)} \left(\ln \mathsf{M}_A(\theta,t) + \rho t - \ln(\rho\varepsilon) \right)
\label{eq:statisticalenvelope}
\end{equation}
is a statistical envelope function of $A(t)$ that provides the sample path guarantee
\begin{equation}
\mathsf{P}[A(\tau,t) \le \mathcal{A}^{\varepsilon}(t-\tau),\, \forall \tau \in [0,t] ] \ge 1-\varepsilon
\label{eq:statisticalenvelopedefinition}
\end{equation}
for all $t \ge 0$. Above, $\varepsilon \in (0,1]$ is a probability of overflow, and $\theta(t) > 0$ and $\rho \in (0,1/\varepsilon]$ are free parameters\footnote{Compared to related works, we use a time-variant parameter $\theta(t)$ instead of a constant $\theta$. This allows optimizing $\theta(t)$ to minimize $\mathcal{A}^{\varepsilon}(t)$ individually for each $t$, which facilitates a computationally efficient implementation.}. With~\eqref{eq:statisticalenvelopedefinition}, statistical performance bounds follow readily by substitution of $\mathcal{A}^{\varepsilon}(t-\tau)$ for $A(\tau,t)$, e.g., the backlog bound~\eqref{eq:backlogbound} yields
\begin{equation}
\mathsf{P} \biggl[B(t) \le \sup_{\tau \in [0,t]} \{\mathcal{A}^{\varepsilon}(t-\tau) - S(\tau,t)\}\biggr] \ge 1-\varepsilon .
\label{eq:statisticalbacklogboundhalf}
\end{equation}

In Fig.~\ref{fig:backlog_qt}, we illustrate the progression of the backlog bound~\eqref{eq:statisticalbacklogboundhalf} over time. The parameters of the service process~\eqref{eq:nonstationarylatencyrate} are $T=100$ and $R=1$. For the arrival process~\eqref{eq:poissonmgf} we use $\alpha=0.09$ and $\beta=0.3$ corresponding to a utilization of $0.3$. We choose $\varepsilon=10^{-9}$ and optimize the free parameters $\theta > 0$ and $\rho \in (0,1/\varepsilon]$ of~\eqref{eq:statisticalenvelope} numerically. Also, we include a backlog bound from the deterministic network calculus, where the service $S_{inv}(\tau,t) = R [t-\tau-T]_+ = S_{inv}(t-\tau)$ is a univariate time-invariant function that depends only on the width of the interval $(\tau,t]$. Otherwise, it considers the worst-case, that is attained for $\tau=0$ where $S_{inv}(t) = S(0,t)$ from~\eqref{eq:nonstationarylatencyrate}. Fig.~\ref{fig:backlog_qt} shows that both, the time-variant and the time-invariant service model, reveal the same growth of the backlog bound until service starts at $T=100$. How the transient backlog is cleared after $T$ and eventually converges to a stationary backlog bound is, however, only explained by the time-variant model. To see why the time-invariant model fails note that the $\sup$ in~\eqref{eq:statisticalbacklogboundhalf} is non-decreasing in $t$ if $S$ is a univariate function.
\begin{figure}
\centering
\includegraphics[width=0.51\columnwidth]{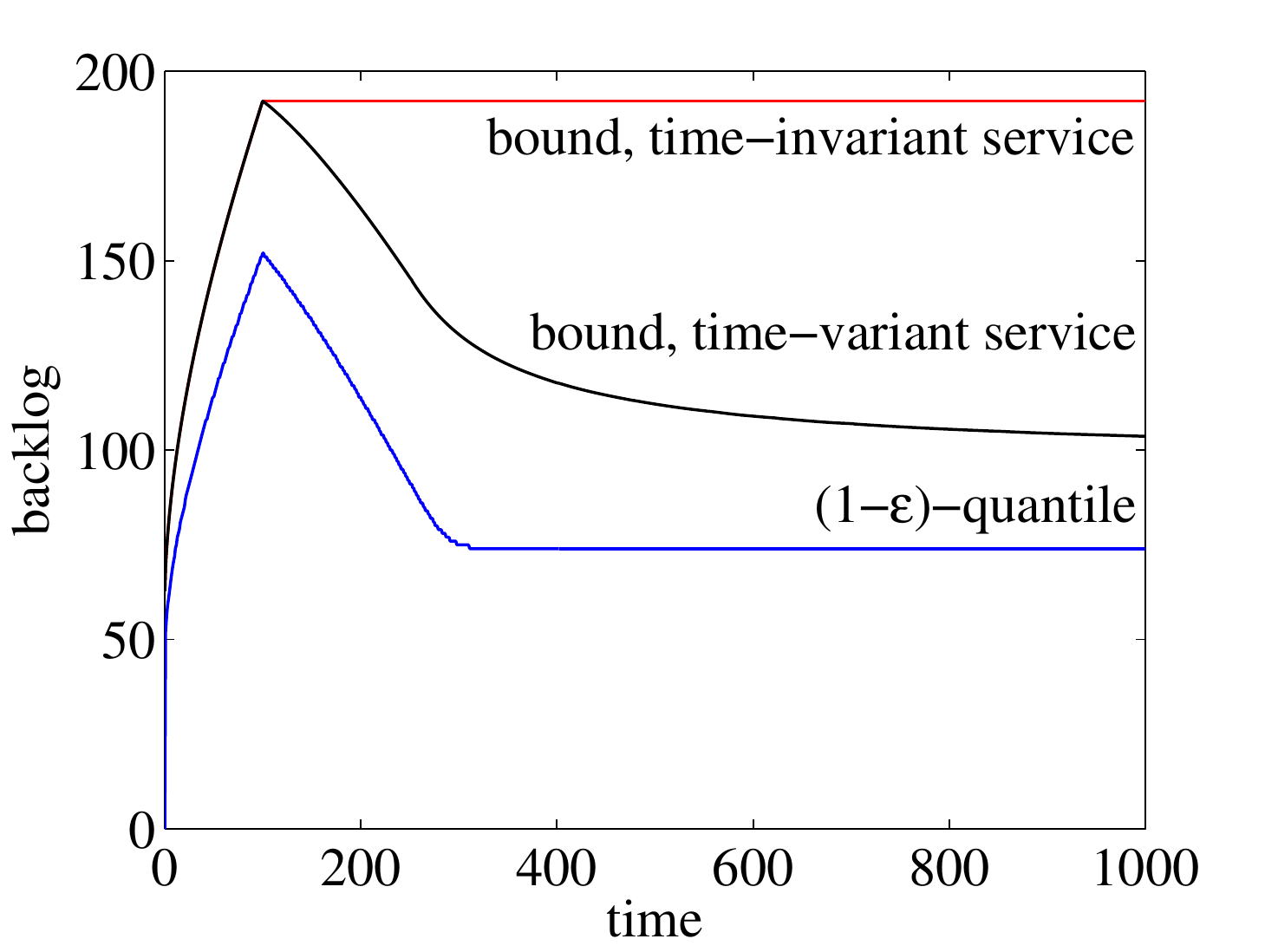}
\vspace{-4pt}
\caption{Progression of the transient backlog over time. The time-variant service model correctly estimates the shape of the $(1-\varepsilon)$-quantile. }
\label{fig:backlog_qt}
\vspace{-10pt}
\end{figure}

Further, Fig.~\ref{fig:backlog_qt} includes the exact backlog quantile for comparison. For the discrete time Poisson model, the solution can be readily obtained from a discrete time Markov chain. The state of the Markov chain $K(t) \ge 0$ represents the number of arrivals that are in the system at $t$. For $t \le T$ the transition matrix $\mathbf{Q}(t)$ is composed of the probabilities $q_{i,i} = 1-\alpha$, $q_{i,i+1} = \alpha$, and all other $q_{i,j} = 0$. For $t > T$ the service starts so that $q_{i,i-1} = (1-\alpha)\beta$, $q_{i,i} = (1-\alpha)(1-\beta) + \alpha\beta$, $q_{i,i+1} = \alpha(1-\beta)$, and all other $q_{i,j} = 0$. The Markov chain starts in state $K(0) = 0$, i.e., the initial state distribution $\mathbf{P}(0)$ is the column vector $(1,0,0,\dots)$. The state distribution for $t > 0$ follows by repeated insertion of $\mathbf{P}(t) = \mathbf{Q}(t) \mathbf{P}(t-1)$.

Clearly, for $t \le T$ the state distribution is binomial, whereas for $t > T$ the distribution makes a transition and for $t \rightarrow \infty$ attains the geometric stationary state distribution~\cite{arita:discretemm1}
\begin{equation*}
\mathsf{P}[K(\infty)=k] = \frac{\beta - \alpha}{\beta (1-\alpha)} \left(\frac{\alpha(1-\beta)}{(1-\alpha)\beta}\right)^k .
\end{equation*}
The backlog distribution can be computed as
\begin{equation*}
\mathsf{P}[B(t)=b] = \sum_{k=1}^{\infty} \mathsf{P}[B(t)=b|K(t)=k] \, \mathsf{P}[K(t)=k]
\end{equation*}
for $b > 0$ and $\mathsf{P}[B(t)=0]=\mathsf{P}[K(t)=0]$ for $b=0$. The conditional backlog in state $k$ is the sum of $k$ geometric random variables that is negative binomial, i.e., for $k,b > 0$
\begin{equation*}
\mathsf{P}[B(t)=b|K(t)=k] = \binom{b-1}{k-1} \beta^{k-1} (1-\beta)^{b-k} .
\end{equation*}

The $(1-\varepsilon)$-quantile of $B(t)$ is depicted in Fig.~\ref{fig:backlog_qt}. We observe that the bound from the time-variant service model provides a good estimate that recovers the shape of the quantile. We note that the deviation is due to bounds that are invoked in the derivation of the Poisson envelope. While the network calculus literature includes envelopes for non-trivial arrival processes including self-similar, long-range dependent~\cite{rizk:fbm,liebeherr:heavytailed}, and heavy-tailed processes~\cite{liebeherr:heavytailed}, tighter martingale bounds are available, e.g., for Poisson~\cite{ciucu:exactnetworkcalculus}, Markov, and autoregressive processes~\cite{chang:performanceguarantees}.

Fig.~\ref{fig:parameters_qt} presents the impact of the arrival rate $\alpha$ for $T=100$ and the impact of the sleep cycle $T$ for $\alpha=0.15$ on the backlog quantile. The remaining parameters are as in Fig.~\ref{fig:backlog_qt}. The measures of interest~\cite{wang:transientatm} are the maximum overshoot compared to the steady-state and the relaxation time, i.e., the time that is required to reach within a defined range the steady-state backlog. Fig.~\ref{fig:load_qt} shows that $\alpha$ has a significant impact on both quantities. Interestingly, if $\alpha$ is large, the maximum overshoot occurs after $T$, i.e., during the transition from binomial to geometric state distribution. The relaxation time reaches values that are larger than $T$ by an order of magnitude.
\begin{figure}
\hspace{-10pt}
\subfigure[$\alpha = \{0.09,0.12,\dots,0.21\}$]{
\includegraphics[width=0.51\columnwidth]{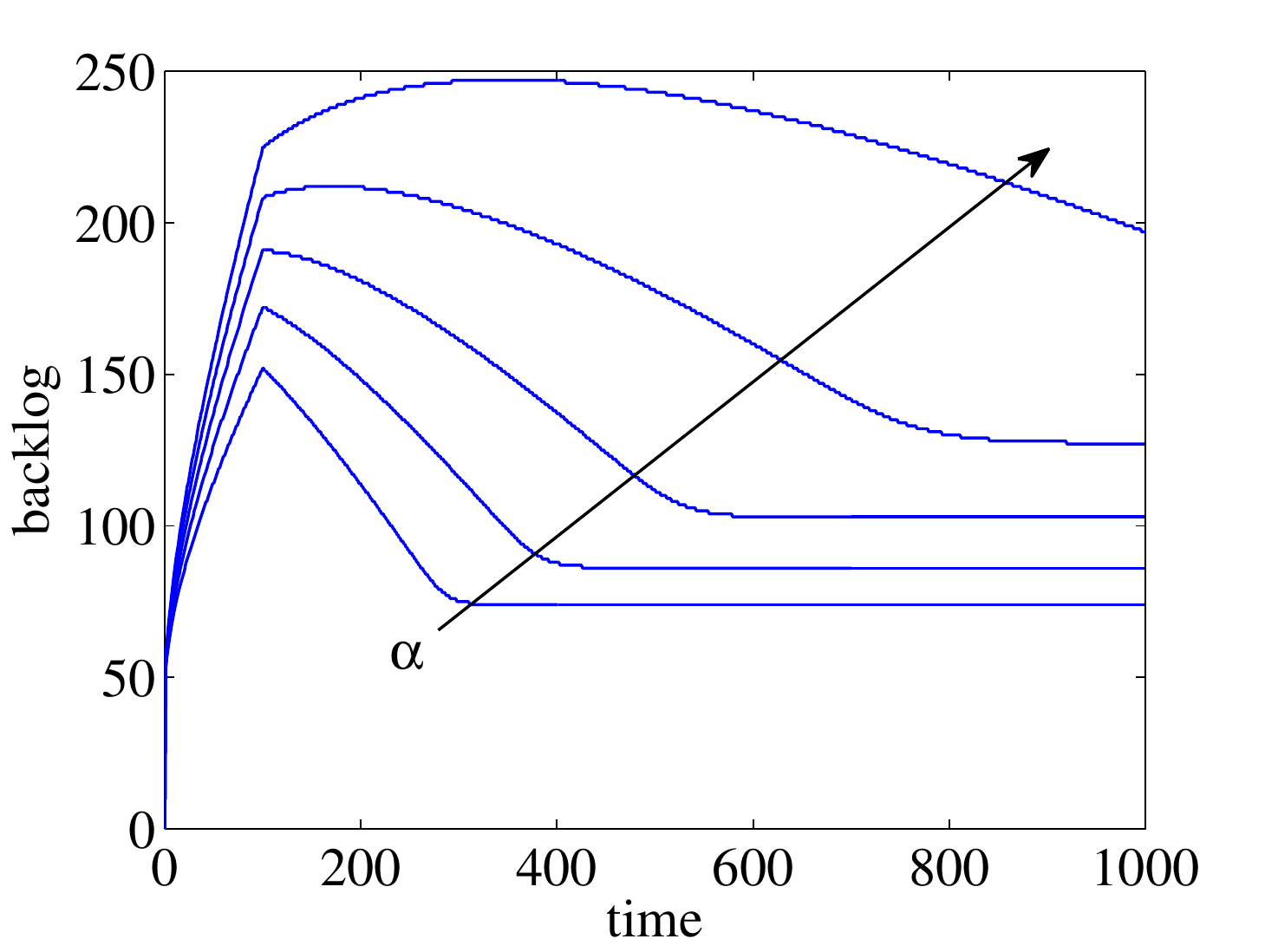}
\label{fig:load_qt}
}
\hspace{-10pt}
\subfigure[$T=\{0,25,\dots,150\}$]{
\includegraphics[width=0.51\columnwidth]{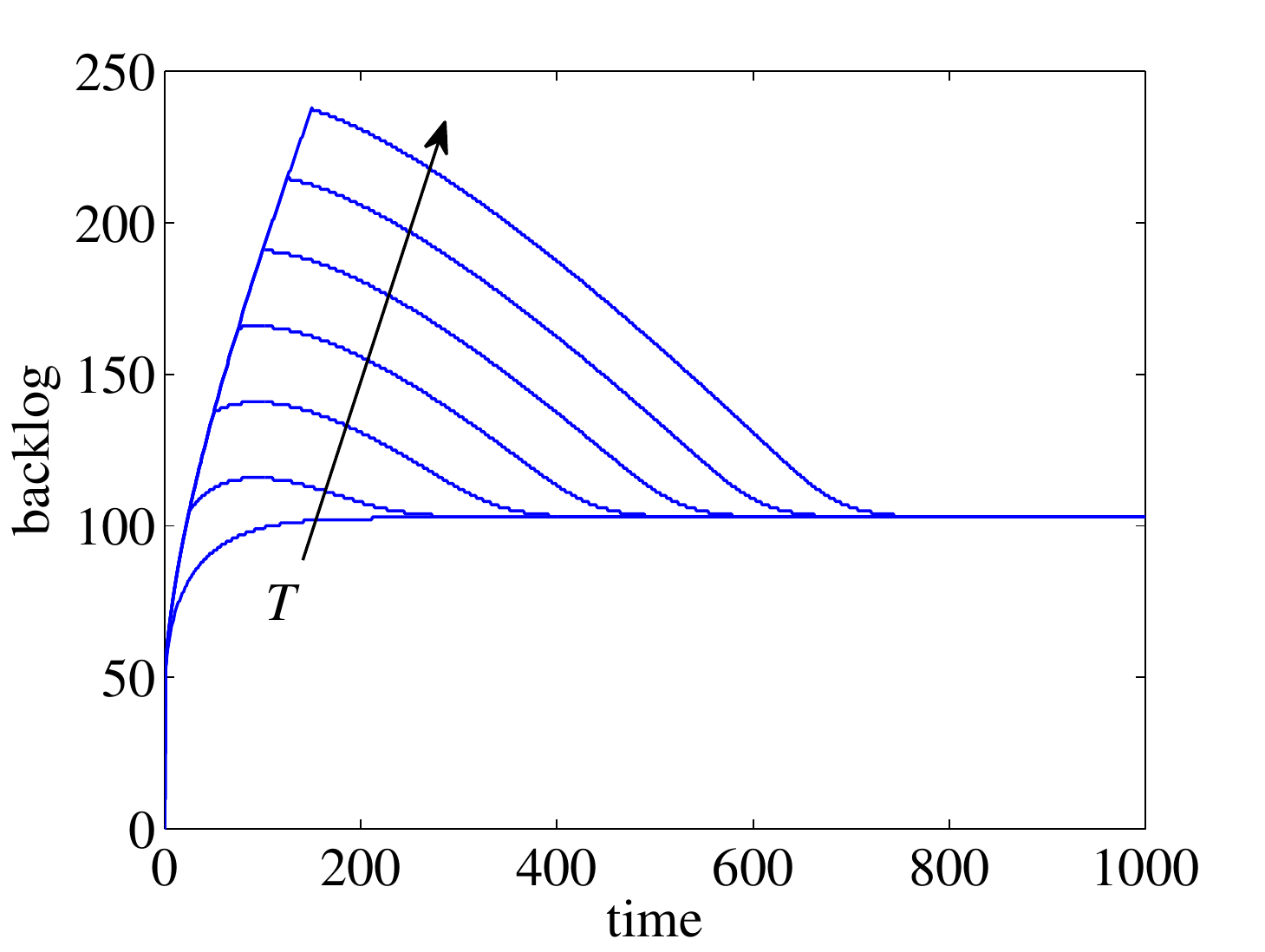}
\label{fig:convergence_qt}
\hspace{-10pt}
}
\vspace{-10pt}
\caption{Impact of arrival rate $\alpha$ and sleep cycle $T$ on the backlog quantile.}
\label{fig:parameters_qt}
\vspace{-10pt}
\end{figure}
%
%
\subsection{Non-stationary Service Curves}
\label{sec:systemmodelrandom}
Next, we consider $S(\tau,t)$ as a non-stationary random service process. We define a bivariate envelope function $\mathcal{S}^{\varepsilon}(\tau,t)$ that conforms to
\begin{equation}
\mathsf{P}[S(\tau,t) \ge \mathcal{S}^{\varepsilon}(\tau,t),\, \forall \tau \in [0,t]] \ge 1-\varepsilon
\label{eq:bivariateenvelope}
\end{equation}
for all $t \ge 0$ where $\varepsilon \in (0,1]$ is a probability of underflow. Adding $A(\tau)$ to both sides we have
\begin{equation}
\mathsf{P}[A(\tau) + S(\tau,t) \ge A(\tau) + \mathcal{S}^{\varepsilon}(\tau,t),\, \forall \tau \in [0,t]] \ge 1-\varepsilon .
\label{eq:bivariateenvelope2}
\end{equation}
Since~\eqref{eq:bivariateenvelope2} makes a sample path argument for all $\tau \in [0,t]$, it follows that
\begin{equation*}
\mathsf{P}\biggl[\inf_{\tau \in [0,t]} \{A(\tau) + S(\tau,t)\} \ge \inf_{\tau \in [0,t]} \{A(\tau) + \mathcal{S}^{\varepsilon}(\tau,t)\}\biggr] \ge 1-\varepsilon .
\end{equation*}
Using~\eqref{eq:serviceprocess} we conclude that
\begin{equation}
\mathsf{P} [D(t) \ge A \otimes \mathcal{S}^{\varepsilon}(t)] \ge 1-\varepsilon.
\label{eq:servicecurve}
\end{equation}
We refer to $\mathcal{S}^{\varepsilon}(\tau,t)$ as non-stationary service curve. It extends the notion of effective service curve~\cite{burchard:endtoendstatisticalcalculus} that defines~\eqref{eq:servicecurve} for univariate functions $\mathcal{S}^{\varepsilon}(t-\tau)$. Compared to~\cite{burchard:endtoendstatisticalcalculus}, the definition of a bivariate function $\mathcal{S}^{\varepsilon}(\tau,t)$ provides a service guarantee that has the capability to consider transient changes over time.

Next, we use the negative MGF, respectively, Laplace transform of $S(\tau,t)$ denoted $\mathsf{M}_S(-\theta,\tau,t) = \mathsf{E}[e^{-\theta S(\tau,t)}]$ to derive a non-stationary service curve. We show that the function
\begin{equation}
\mathcal{S}^{\varepsilon}(\tau,t) = -\frac{1}{\theta(\tau,t)} ( \ln \mathsf{M}_S(-\theta,\tau,t) + \rho (t-\tau) - \ln(\rho\varepsilon)\!) ,
\label{eq:statisticalserviceenvelope}
\end{equation}
where $\theta(\tau,t) > 0$ and $\rho \in (0,1/\varepsilon]$ are free parameters, satisfies the service curve guarantee~\eqref{eq:servicecurve}.

For completeness, we include the derivation of~\eqref{eq:statisticalserviceenvelope} that extends~\cite{Fidler:2014:CDE} to non-stationary processes. The derivation employs basic steps from the stochastic network calculus~\cite{li:effectivebandwidthcalculus2, ciucu:networkservicecurvescaling2}. We use the complementary formulation of~\eqref{eq:bivariateenvelope}
\begin{equation*}
\xi := \mathsf{P}[\exists \tau \in [0,t]:  S(\tau,t) < \mathcal{S}^{\varepsilon}(\tau,t)] \le \varepsilon .
\end{equation*}
to prove that $\mathcal{S}^{\varepsilon}(\tau,t)$ defined in~\eqref{eq:statisticalserviceenvelope} satisfies~\eqref{eq:bivariateenvelope}, or equivalently that $\xi \le \varepsilon$. Using the union bound and Chernoff's lower bound $\mathsf{P}[X \le x] \le e^{\theta x} \mathsf{M}_X(-\theta)$ for $\theta \ge 0$ it holds that
\begin{equation*}
\xi \le \! \sum_{\tau=0}^{t-1} \mathsf{P}[S(\tau,t) < \mathcal{S}^{\varepsilon}(\tau,t)] \le \! \sum_{\tau=0}^{t-1} e^{\theta(\tau,t) \mathcal{S}^{\varepsilon}(\tau,t)} \mathsf{M}_S(-\theta,\tau,t) ,
\end{equation*}
where $\theta(\tau,t) \ge 0$ is a set of free parameters. The case where $\tau=t$ is omitted since $S(t,t) = 0$ and $\mathcal{S}^{\varepsilon}(t,t) \le 0$ by definition. By insertion of $\mathcal{S}^{\varepsilon}(\tau,t)$ from~\eqref{eq:statisticalserviceenvelope} it follows that
\begin{equation*}
\xi \le \rho\varepsilon \! \sum_{\tau=0}^{t-1} \! e^{-\rho(t-\tau)} = \rho\varepsilon \! \sum_{\upsilon=1}^{t} \! e^{-\rho\upsilon} \le \rho\varepsilon \! \int_{0}^{\infty} \!\!\! e^{-\rho y} dy = \varepsilon
\end{equation*}
where each summand is bounded by $e^{-\rho\upsilon} \le \int_{\upsilon-1}^{\upsilon} e^{-\rho y} dy$ since $e^{-\rho\upsilon}$ is decreasing. Finally, letting $t \rightarrow \infty$ and solving the integral completes the proof that $\xi \le \varepsilon$ for all $t \ge 0$.
%
%
\vspace{5pt}
\subsubsection*{Random Sleep Scheduling}
We extend the deterministic sleep scheduling model from Sec.~\ref{sec:systemmodelregenerative} and consider a non-stationary work-conserving system with random service increments $Z(t)$ for $t \ge 0$. When entering sleep state, the system regenerates and wakes up after a random time $T \ge 0$, i.e., $Z(t)=0$ for $t \in [0,T]$. The service process is computed as $S(\tau,t) = \sum_{\upsilon=\tau+1}^t Z(\upsilon)$ for all $t > \tau \ge 0$ and $S(t,t)=0$ for all $t \ge 0$~\cite{chang:performanceguarantees}. To derive the MGF of $S(\tau,t)$, we first consider the number of usable time-slots in $(\tau,t]$, i.e., after time $T$
\begin{equation*}
U(\tau,t) = [t-\max\{\tau,T\}]_+ .
\end{equation*}
The MGF of $U(\tau,t)$ is composed of three terms
\begin{multline}
\mathsf{M}_U(\theta,\tau,t) = \\ e^{\theta (t-\tau)} \mathsf{P}[T \le \tau] + \sum_{\upsilon = \tau+1}^{t} e^{\theta(t-\upsilon)} \mathsf{P}[T = \upsilon] + \mathsf{P}[T > t] ,
\label{eq:slotsafterT}
\end{multline}
that correspond to the cases where the start of the service $T$ occurs before and including $\tau$, within $(\tau,t]$, and after $t$, respectively. Given the service increments $Z(t)$ for $t > T$ are iid with MGF $\mathsf{M}_Z(\theta)$, the MGF of the service process is~\cite{ross:probability}
\begin{equation}
\mathsf{M}_S(\theta,\tau,t) \!=\! \mathsf{E}\bigl[(\mathsf{M}_Z(\theta))^{U(\tau,t)}\bigr] \!=\! \mathsf{M}_U(\ln \mathsf{M}_Z(\theta),\tau,t) .
\label{eq:servicemgf}
\end{equation}

For a concrete example, we model $T$ as a geometric random variable with parameter $p$, where $\mathsf{P}[T = \upsilon] = p(1-p)^\upsilon$, and $Z(t)$ for $t > T$ as iid Bernoulli trials with parameter $q$. Due to the memorylessness of the processes, solutions for this specific example may also be derived, e.g., from a Markov model. We use this example as it enables us to compute certain reference results in Sec.~\ref{sec:estimation}. Note that the service curve~\eqref{eq:statisticalserviceenvelope} in general is not limited to memoryless processes.

Regarding~\eqref{eq:slotsafterT} we have $\mathsf{P}[T \le \tau] = 1-(1-p)^{\tau+1}$ and
\begin{equation*}
\sum_{\upsilon = \tau+1}^{t} e^{\theta(t-\upsilon)} \mathsf{P}[T = \upsilon] = e^{\theta t} p \sum_{\upsilon = \tau+1}^{t} \bigl(e^{-\theta} (1-p)\bigr)^\upsilon ,
\end{equation*}
where we substitute $y = e^{-\theta} (1-p)$ and compute
\begin{equation*}
\sum_{\upsilon = \tau+1}^{t} y^{\upsilon} = \frac{y^{\tau+1}-y^{t+1}}{1-y}.
\end{equation*}
Having obtained a solution of~\eqref{eq:slotsafterT}, the MGF of the service process follows from~\eqref{eq:servicemgf} by insertion of $\mathsf{M}_Z(\theta) = q e^{\theta} + 1-q$ for a Bernoulli service increment process. Finally, the non-stationary service curve is computed from~\eqref{eq:statisticalserviceenvelope}.

\begin{figure}
\hspace{-10pt}
\subfigure[$\tau = \{0,100,\dots\}$]{
\includegraphics[width=0.51\columnwidth]{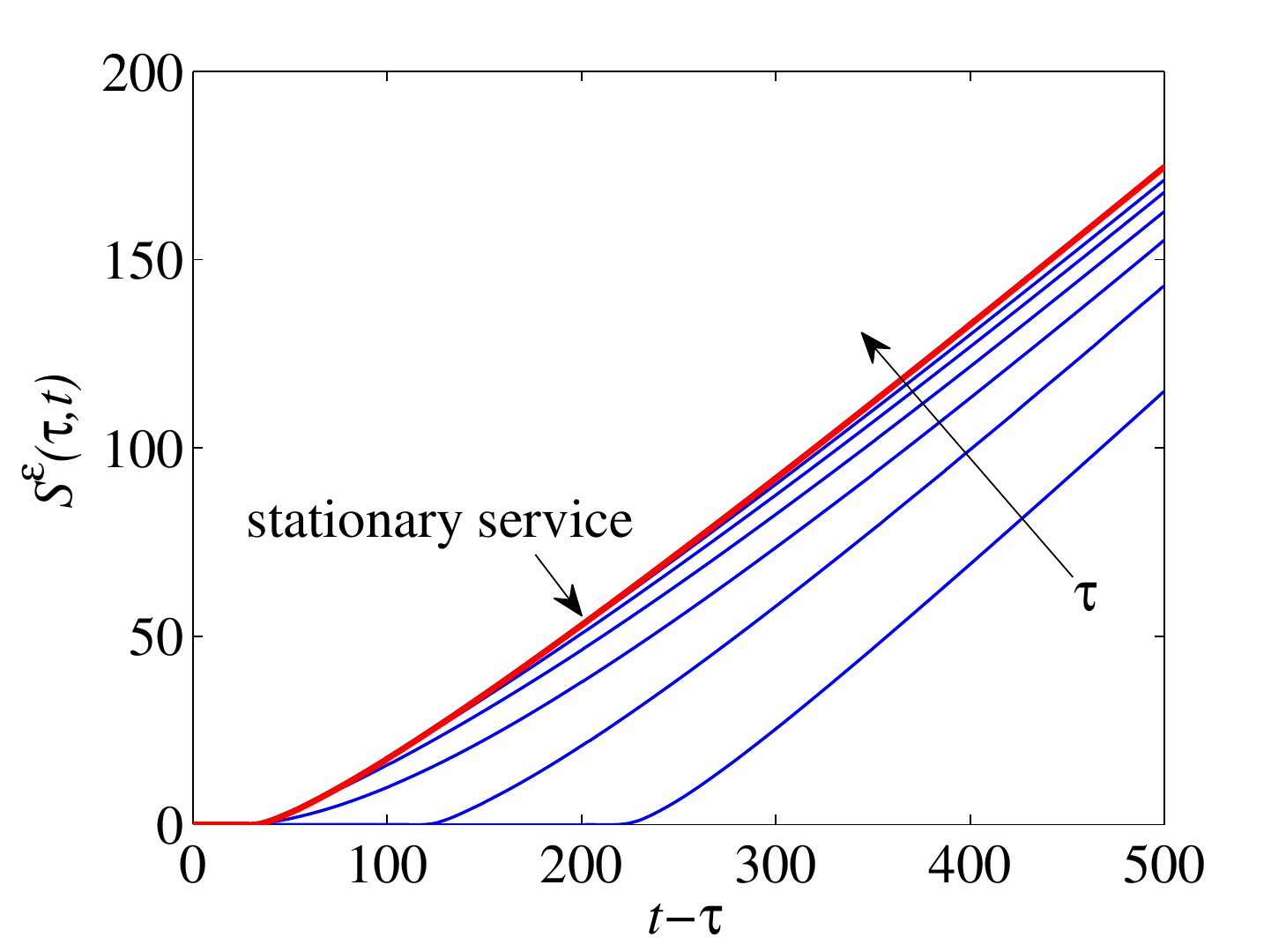}
\label{fig:servicecurvesforward}
}
\hspace{-10pt}
\subfigure[$t = \{0,100,\dots\}$]{
\includegraphics[width=0.51\columnwidth]{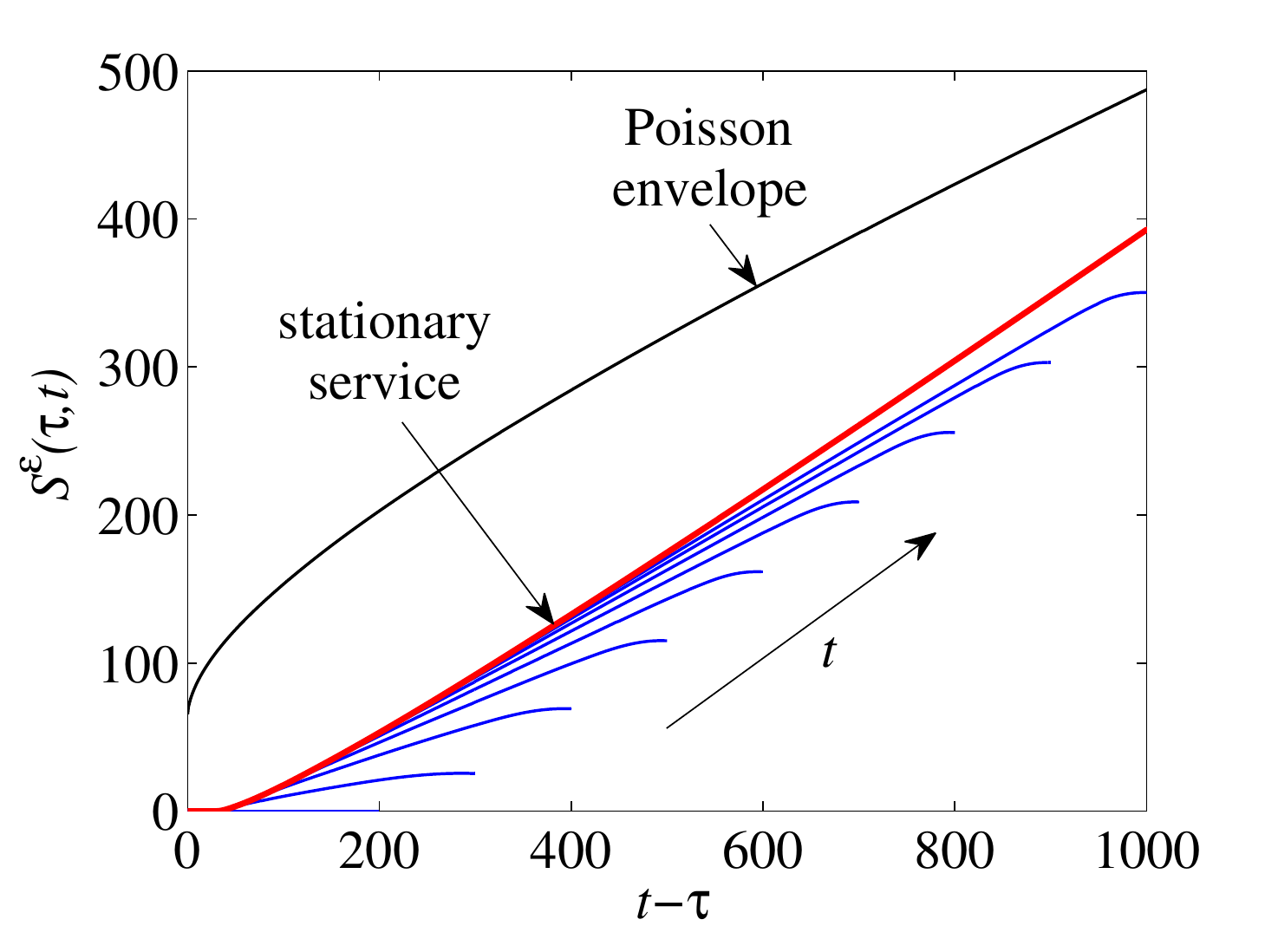}
\label{fig:servicecurvesbackwards}
\hspace{-10pt}
}
\vspace{-10pt}
\caption{Non-stationary service curves of random sleep scheduling.}
\label{fig:servicecurves}
\vspace{-10pt}
\end{figure}
Fig.~\ref{fig:servicecurves} illustrates $\mathcal{S}^{\varepsilon}(\tau,t)$ for $p=0.1$ and $q=0.5$. The remaining parameters of the service curve are $\varepsilon = 10^{-6}$, $\rho = 10^{-4}$, and $\theta(\tau,t)$ is optimized numerically. In Fig.~\ref{fig:servicecurvesforward}, we show how the service in an interval of width $t-\tau$ increases with increasing distance $\tau$ from the last regeneration point. For small $\tau$ a significant impact of the initial transient phase is noticed. For large $\tau$ we observe that $\mathcal{S}^{\varepsilon}(\tau,t)$ converges towards a stationary service curve that is computed for a Bernoulli service increment process without sleep scheduling. Fig.~\ref{fig:servicecurvesbackwards} displays service curves for fixed $t$ and variable $\tau$. Here, the transient phase is reflected in the non-convex shape of the curves, rightwards where $\tau$ approaches zero. In contrast, the initial delay at the origin, that also applies to the stationary service curve, is caused by the Bernoulli service increment process. For small intervals $t-\tau$ the process results in a service of zero with non-negligible probability. For $t = 100$ and for $t=200$ the effects bring about a service curve of zero.

The presentation of $\mathcal{S}^{\varepsilon}(\tau,t)$ in Fig.~\ref{fig:servicecurvesbackwards} conforms with the formulation of statistical performance bounds, where $t$ is fixed and all $\tau \in [0,t]$ are evaluated. A statistical backlog bound follows from~\eqref{eq:servicecurve} with~\eqref{eq:statisticalenvelopedefinition} as
\begin{equation}
\mathsf{P} \Bigl[B(t) \le \sup_{\tau \in [0,t]} \{\mathcal{A}^{\varepsilon}(t-\tau) - \mathcal{S}^{\varepsilon}(\tau,t)\}\Bigr] \ge 1-2\varepsilon
\label{eq:statisticalbacklogbound}
\end{equation}
and a first-come first-served delay bound as
\begin{multline}
\!\!\!\!\!\mathsf{P} \Bigl[W(t) \le \inf \Bigl\{w \! \ge \! 0: \! \sup_{\tau \in [0,t]} \{\mathcal{A}^{\varepsilon}(t-\tau) - \mathcal{S}^{\varepsilon}(\tau,t+w)\} \le 0 \Bigr\}\Bigr] \\ \ge 1-2\varepsilon .
\label{eq:statisticaldelaybound}
\end{multline}
Intuitively, the backlog and delay bound are the maximal vertical and horizontal deviation of $\mathcal{A}^{\varepsilon}(t-\tau)$ and $\mathcal{S}^{\varepsilon}(\tau,t)$, respectively. For illustration, Fig.~\ref{fig:servicecurvesbackwards} includes the Poisson arrival envelope from Sec.~\ref{sec:systemmodelregenerative} for $\alpha = 0.06$, $\beta = 0.3$, and $\varepsilon = 10^{-6}$. Corresponding backlog and delay bounds are shown in Fig.~\ref{fig:backlogdelay} for different parameters of the service $p$ and $q$ where the long-term utilization is $\alpha/(\beta q)$. The remaining parameters $\theta$ and $\rho$ are optimized numerically. Compared to the stationary case, the transient overshoot of sleep scheduling is considerable.
\begin{figure}
\hspace{-10pt}
\includegraphics[width=0.51\columnwidth]{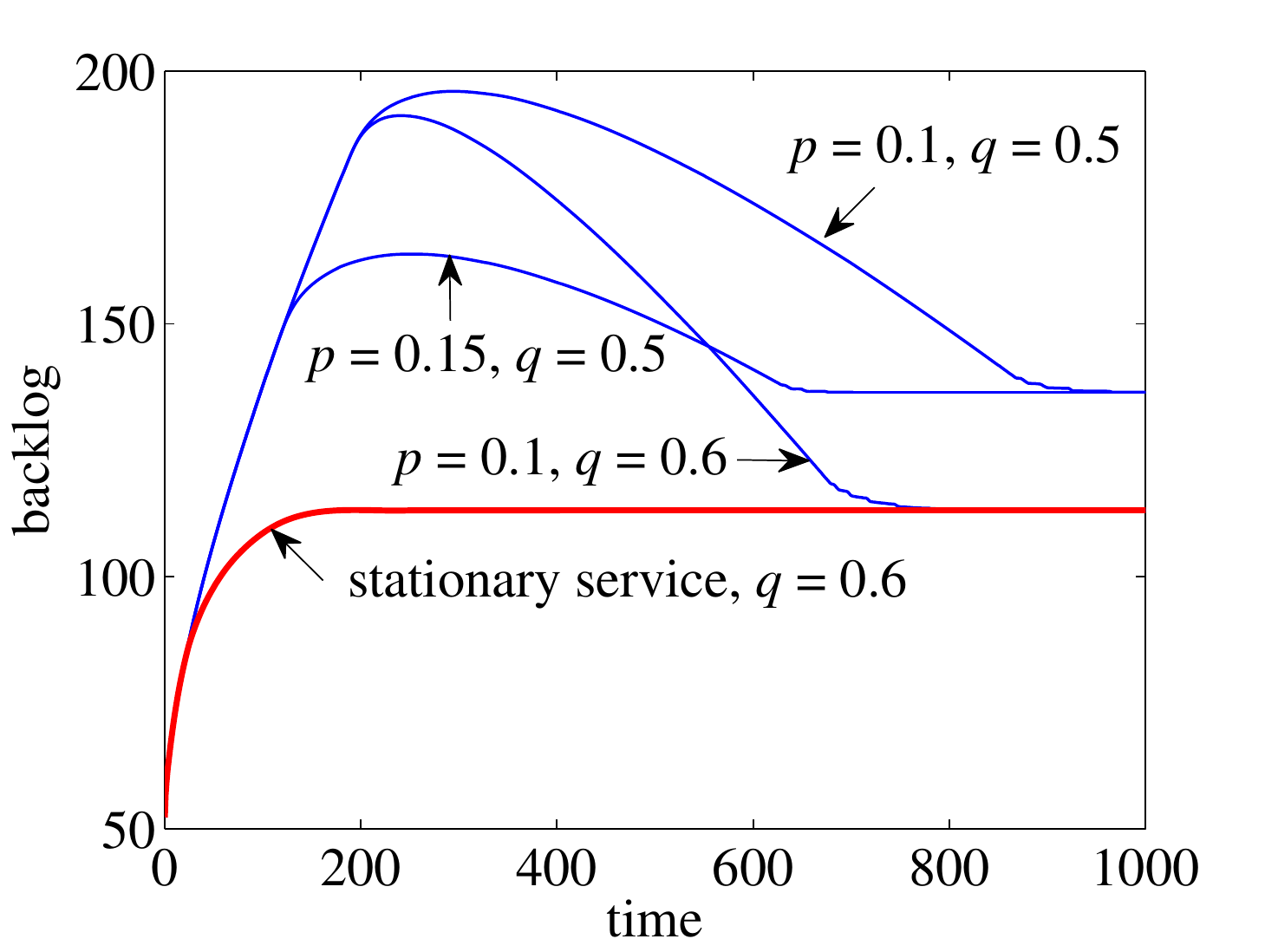}
\hspace{-10pt}
\includegraphics[width=0.51\columnwidth]{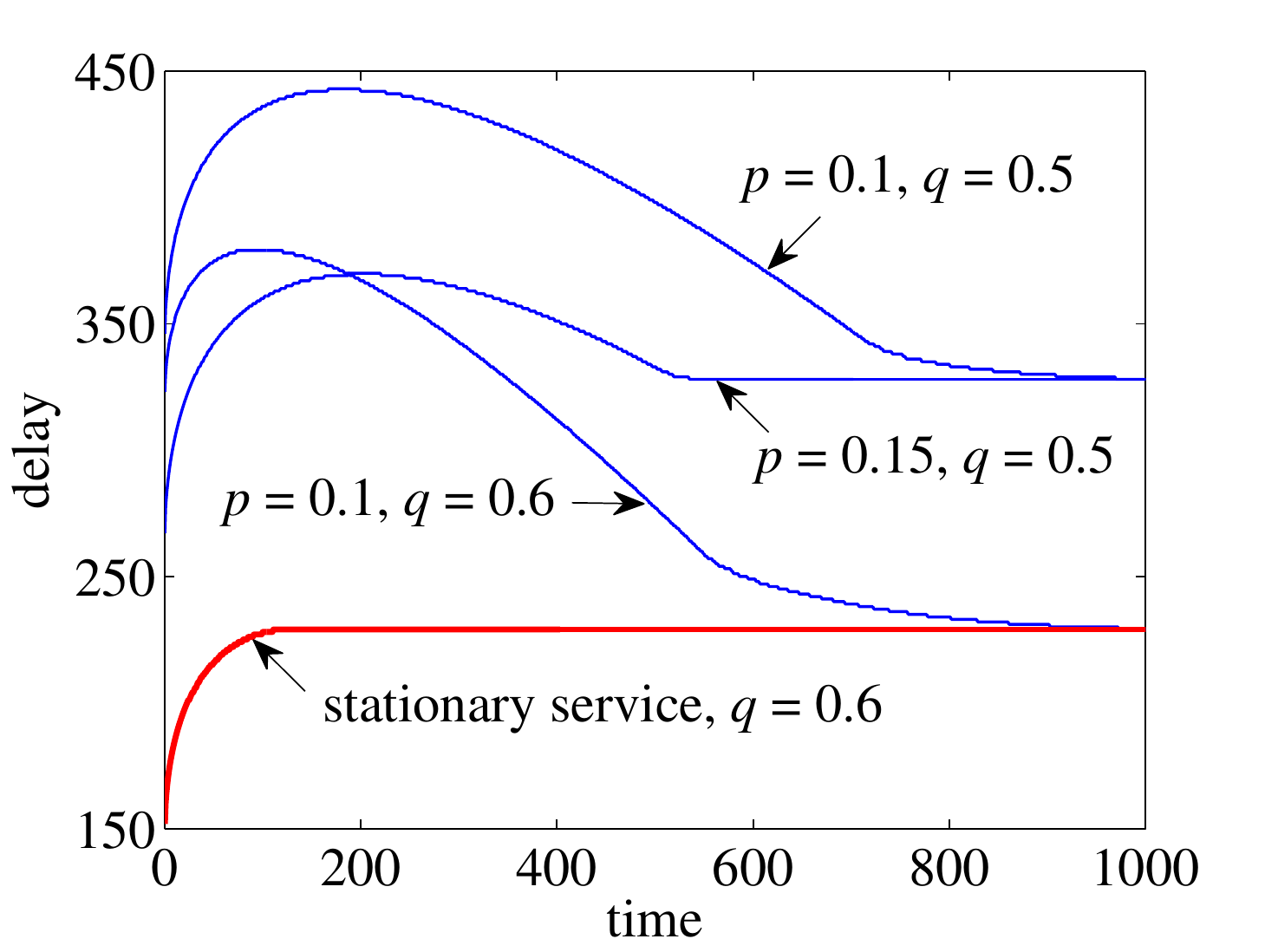}
\hspace{-10pt}
\vspace{-4pt}
\caption{Transient backlog and delay of random sleep scheduling.}
\label{fig:backlogdelay}
\vspace{-10pt}
\end{figure}
%
%
\section{Measurement-based Estimation}
\label{sec:estimation}
Complementary to the model-based approach, methods for estimation of service curves from measurements of systems have received increasing attention in the past years~\cite{cetinkaya:egressadmissioncontrol2, valaee:adhocadmissioncontrol, undheim:netcalcroutermeasurements, bredel:netcalcroutermeasurements, hisakado:legendre, liebeherr:availbw, rizk:identifiability, luebben:availbw2, luebben:availbwdiss}. A closely related research area~\cite{liebeherr:availbw} is available bandwidth estimation that seeks to estimate the long-term average unused capacity of a system or a network. In passive measurements~\cite{cetinkaya:egressadmissioncontrol2, liebeherr:availbw} the departures of a system are observed given the production traffic of the network, whereas in active probing~\cite{valaee:adhocadmissioncontrol, undheim:netcalcroutermeasurements, bredel:netcalcroutermeasurements, hisakado:legendre, liebeherr:availbw, rizk:identifiability, luebben:availbw2, luebben:availbwdiss} artificial test traffic is injected into the system. The different approaches can be further classified to be based either on the assumption of a work-conserving system~\cite{cetinkaya:egressadmissioncontrol2, valaee:adhocadmissioncontrol, undheim:netcalcroutermeasurements, bredel:netcalcroutermeasurements} or more general min-plus systems~\cite{hisakado:legendre, liebeherr:availbw, rizk:identifiability, luebben:availbw2,luebben:availbwdiss} and either use a deterministic model~\cite{undheim:netcalcroutermeasurements, bredel:netcalcroutermeasurements, hisakado:legendre, liebeherr:availbw, rizk:identifiability} or a stochastic model~\cite{cetinkaya:egressadmissioncontrol2, valaee:adhocadmissioncontrol, luebben:availbw2, luebben:availbwdiss}. For a more detailed comparison see also~\cite{fidler:netcalcsurvey}. While a method known as rate scanning is available for estimation of stationary service curves for the general class of min-plus linear systems with random service~\cite{luebben:availbw2, luebben:availbwdiss}, methods for estimation of transient service curves are not elaborated in the current literature.

Like~\cite{hisakado:legendre, liebeherr:availbw, rizk:identifiability, luebben:availbw2, luebben:availbwdiss}, we consider min-plus linear systems, where~\eqref{eq:serviceprocess} holds with equality, i.e., for all $t \ge 0$
\begin{equation}
D(t) = \inf_{\tau \in [0,t]} \{ A(\tau) + S(\tau,t) \} .
\label{eq:linearsystem}
\end{equation}
The unknown $S(\tau,t)$ is a time-variant, random service process as in~\cite{luebben:availbw2, luebben:availbwdiss}. Compared to~\cite{luebben:availbw2, luebben:availbwdiss} we do, however, not assume stationarity of $S(\tau,t)$. Based on~\eqref{eq:linearsystem}, the goal of service curve estimation can be phrased as an inversion problem, i.e., solving~\eqref{eq:linearsystem} for $S(\tau,t)$. Due to the infimum, the min-plus convolution has no inverse operation in general. A solution can, however, be obtained for certain arrival functions $A(t)$~\cite{liebeherr:availbw}. Using measurements of $S(\tau,t)$, we finally seek to estimate a maximal non-stationary service curve $\mathcal{S}^{\varepsilon}(\tau,t)$, i.e., find a maximal\footnote{We note that in general there is no unique maximal solution but only Pareto efficient solutions.} function $\mathcal{S}^{\varepsilon}(\tau,t)$ that satisfies~\eqref{eq:servicecurve}. The measurements are performed assuming a regenerative service process as defined in Sec.~\ref{sec:systemmodelregenerative}, where repeated network probes can observe samples of the service process $S_i(\tau,t)$ at regeneration point $P_i$ as defined by~\eqref{eq:regenerativeservicsample}.

The remainder of this section is structured as follows: We adapt two known methods, rate scanning (Sec.~\ref{sec:ratescanning}) and burst response (Sec.~\ref{sec:burstresponse}), for estimation of non-stationary service curves. Fundamental limitations of these methods are identified that are explained by the non-convexity and the super-additivity of the service (Sec.~\ref{sec:superadditivity}), respectively. We devise a new two-phase probing method (Sec.~\ref{sec:minimalprobing}) that first determines the shape of a suitable probe, that is proven to be minimal under certain conditions. The probe is used to obtain a service curve estimate with a defined accuracy.
%
%
\subsection{Rate Scanning}
\label{sec:ratescanning}
First, we consider the rate scanning method from~\cite{liebeherr:availbw, luebben:availbw2} that uses constant rate probes $A(t) = rt$ for a set of rates $r \in \mathbb{R}$. While~\cite{liebeherr:availbw} and~\cite{luebben:availbw2} consider deterministic and stationary service curves, respectively, we adapt the method to provide estimates of non-stationary service curves for transient phases.

For constant rate probes $A(t) = r t$ the backlog of a system follows from~\eqref{eq:linearsystem} as $B(t) = \sup_{\tau \in [0,t]} \{r (t-\tau) - S(\tau,t) \}$. Consequently, it holds that $B(t) \ge r (t-\tau) - S(\tau,t)$ for all $\tau \in [0,t]$. Solving for $S(\tau,t)$ provides the lower bound $S(\tau,t) \ge r (t-\tau) - B(t)$ for all $\tau \in [0,t]$.

Next, we use the definition of $(1-\xi)$-quantile
\begin{equation}
X^{\xi} = \inf \{x \ge 0: \mathsf{P}[X \le x] \ge 1-\xi \}
\label{eq:quantile}
\end{equation}
and denote $B^{\xi}(r,t)$ the backlog quantile at time $t$ as a function of $r$. A non-stationary service curve~\eqref{eq:servicecurve} follows as
\begin{equation}
\mathcal{S}_{rs}^{\varepsilon}(\tau,t) = \max_{r \in \mathbb{R}} \{r (t-\tau) - B^{\xi}(r,t)\}
\label{eq:ratescanestimate}
\end{equation}
for $t \ge \tau \ge 0$. By application of the union bound, $\mathcal{S}_{rs}^{\varepsilon}(\tau,t)$ satisfies~\eqref{eq:bivariateenvelope} with probability $\varepsilon = \sum_{r \in \mathbb{R}} \xi$~\cite{luebben:availbw2}. Compared to~\cite{luebben:availbw2},~\eqref{eq:ratescanestimate} uses the transient instead of the stationary backlog to estimate a non-stationary service curve. The service curve that is defined by~\eqref{eq:ratescanestimate} has the form of a Legendre-Fenchel transform of the backlog~\cite{liebeherr:availbw}. The Legendre-Fenchel transform has a number of useful properties in the network calculus~\cite{hisakado:legendre, fidler:legendre}. It is also known as convex conjugate as the result is generally a convex function~\cite{rockafellar:convexanalysis}. For convex functions, the Legendre-Fenchel transform is its own inverse, whereas the bi-conjugate of a non-convex function can only recover the convex hull of the original function~\cite{rockafellar:convexanalysis}.

In a practical implementation of rate scanning, the set of probing rates $\mathbb{R}$ has to be selected. Options are, e.g., linearly or geometrically spaced rates combined with suitable tests that determine the maximum probing rate~\cite{liebeherr:availbw,luebben:availbw2}. Then, estimates of the backlog quantile $B^{\xi}(r,t)$ for $t \ge 0$ are obtained from repeated measurements for each $r \in \mathbb{R}$. A service curve estimate $\mathcal{S}_{rs}^{\varepsilon}(\tau,t)$ follows from~\eqref{eq:ratescanestimate}.

For evaluation of the method, we consider a simulation of the random sleep scheduling model from Sec.~\eqref{sec:systemmodelrandom} with identical parameters $p=0.1$ and $q=0.5$, i.e., the stationary service rate is $0.5$. We perform rate scanning with ten rates $r \in \{0.05,0.1,\dots,0.5\}$. For each rate we obtain $10^5$ backlog samples from repeated experiments from which we extract an estimate of the backlog quantile $B^{\xi}(r,t)$ for $\xi = 10^{-4}$ so that $\varepsilon = \sum_{r\in\mathbb{R}} \xi = 10^{-3}$. Fig.~\ref{fig:ratescanning} shows the linear segments obtained by each of the ten probing rates $r \in \mathbb{R}$ and the resulting estimate $\mathcal{S}_{rs}^{\varepsilon}(\tau,t)$ for $t=200$.

As a reference, we include an analytical upper bound\footnote{For a Bernoulli service increment process with parameter $q$ that starts after a geometrically distributed time $T$, an upper bound of the service provided in $[\tau,t]$ is derived as the $(1-\varepsilon)$-quantile of a binomial distribution with parameters $t-\max\{\tau,T\}$ and $q$. We note that the first parameter, i.e., the number of trials $t-\max\{\tau,T\}$, is a random variable.} in Fig.~\ref{fig:ratescanning}. Any function that exceeds the upper bound for some $\tau \in [0,t]$ violates the definition of service envelope~\eqref{eq:bivariateenvelope}. Also, we show an analytical service curve\footnote{The analytical service curve again uses the binomial distribution but makes a sample path argument using the union bound in addition.} for comparison. Clearly, the service curve estimate from rate scanning cannot recover the non-convex parts of the analytical results as it is fundamentally limited to a convex hull by construction.
%
%
\subsection{Burst Response}
\label{sec:burstresponse}
Motivated by the min-plus systems theory of the network calculus~\cite{leboudec:networkcalculus}, a canonical probe for system identification is the burst function $\delta(\tau)$ that is defined as
\begin{equation}
\delta(\tau) = \begin{cases} 0 & \text{for } \tau=0, \\ \infty & \text{for } \tau > 0. \end{cases}
\label{eq:burstfunction}
\end{equation}
In min-plus algebra, the burst function takes the role of the Dirac delta function, i.e., it is the neutral element of min-plus convolution. Consequently, sending a burst probe $A(\tau) = \delta(\tau)$ reveals the service $S(0,t)$ for all $t \ge 0$ as the burst response of the system
\begin{equation}
D(t) = \inf_{\tau \in [0,t]} \{ \delta(\tau) + S(\tau,t) \} = S(0,t) .
\label{eq:burstresponse}
\end{equation}
For additive service processes as defined in~\cite[p.~6]{jiang:stochasticnetworkcalculus}, $S(\tau,t)$ can be obtained for all $t \ge \tau \ge 0$ as
\begin{equation}
S(\tau,t) = S(0,t) - S(0,\tau) .
\label{eq:additivity}
\end{equation}

For a stochastic analysis, we denote $\Omega$ the set of all feasible sample paths $D_{\omega}(t)$. For each $\omega \in \Omega$ we use the additivity~\eqref{eq:additivity} to obtain the service process from the burst response~\eqref{eq:burstresponse} as
\begin{equation*}
S_{\omega}(\tau,t) = D_{\omega}(t) - D_{\omega}(\tau)
\end{equation*}
for all $\tau \in [0,t]$. We fix $t > 0$ and select a subset of the sample paths $\Psi_t \subseteq \Omega$ with probability $\mathsf{P}[\Psi_t] \ge 1-\varepsilon$. By definition of
\begin{equation}
\mathcal{S}_{br}^{\varepsilon}(\tau,t) = \inf_{\psi \in \Psi_t} \{S_{\psi}(\tau,t)\} ,
\label{eq:effectiveserviceburstrespones}
\end{equation}
it holds that $S_{\psi}(\tau,t) \ge \mathcal{S}_{br}^{\varepsilon}(\tau,t)$ for all $\tau \in [0,t]$ and all $\psi \in \Psi_t$. Further, we have $\mathsf{P}[\Psi_t] \ge 1-\varepsilon$ so that $\mathcal{S}_{br}^{\varepsilon}(\tau,t)$ satisfies~\eqref{eq:bivariateenvelope}. Hence, $\mathcal{S}_{br}^{\varepsilon}(\tau,t)$ is a non-stationary service curve that conforms to~\eqref{eq:servicecurve}.

In a practical probing scheme, we can only observe a finite set of sample paths $\Omega$ from repeated measurements of the departures $D_{\omega}(t)$ for $t \ge 0$. We fix $t > 0$ and estimate the service $S_{\omega}(\tau,t) = D_{\omega}(t) - D_{\omega}(\tau)$ for all $\tau \in [0,t]$ and $\omega \in \Omega$. To construct $\Psi_t$ we remove a set of minimal sample paths from $\Omega$. We define the minimal sample path $\phi$ to be the sample path that attains the minimum
\begin{equation*}
S_{\min}(\tau,t) = \min_{\omega \in \Omega} \{S_{\omega}(\tau,t)\}
\vspace{-1pt}
\end{equation*}
for $\tau \in [0,t]$ most frequently. That is $\phi = \arg\max_{\omega \in \Omega} \{X_{\omega} \}$, where for all $\omega \in \Omega$
\begin{equation*}
X_{\omega} = \sum_{\tau=0}^{t-1} 1_{S_{\omega}(\tau,t) = S_{\min}(\tau,t)} .
\vspace{-1pt}
\end{equation*}
The indicator function $1_{(\cdot)}$ is one if the argument is true and zero otherwise. We remove $\phi$ to obtain $\Psi_t = \Omega \backslash \phi$ and repeat the above steps for $\Psi_t$ as long as $\mathsf{P}[\Psi_t] \ge 1-\varepsilon$. Finally, we obtain the service curve estimate from~\eqref{eq:effectiveserviceburstrespones} for all $\tau \in [0,t]$.

\begin{figure}
\hspace{-10pt}
\subfigure[Rate scanning]{
\includegraphics[width=0.51\columnwidth]{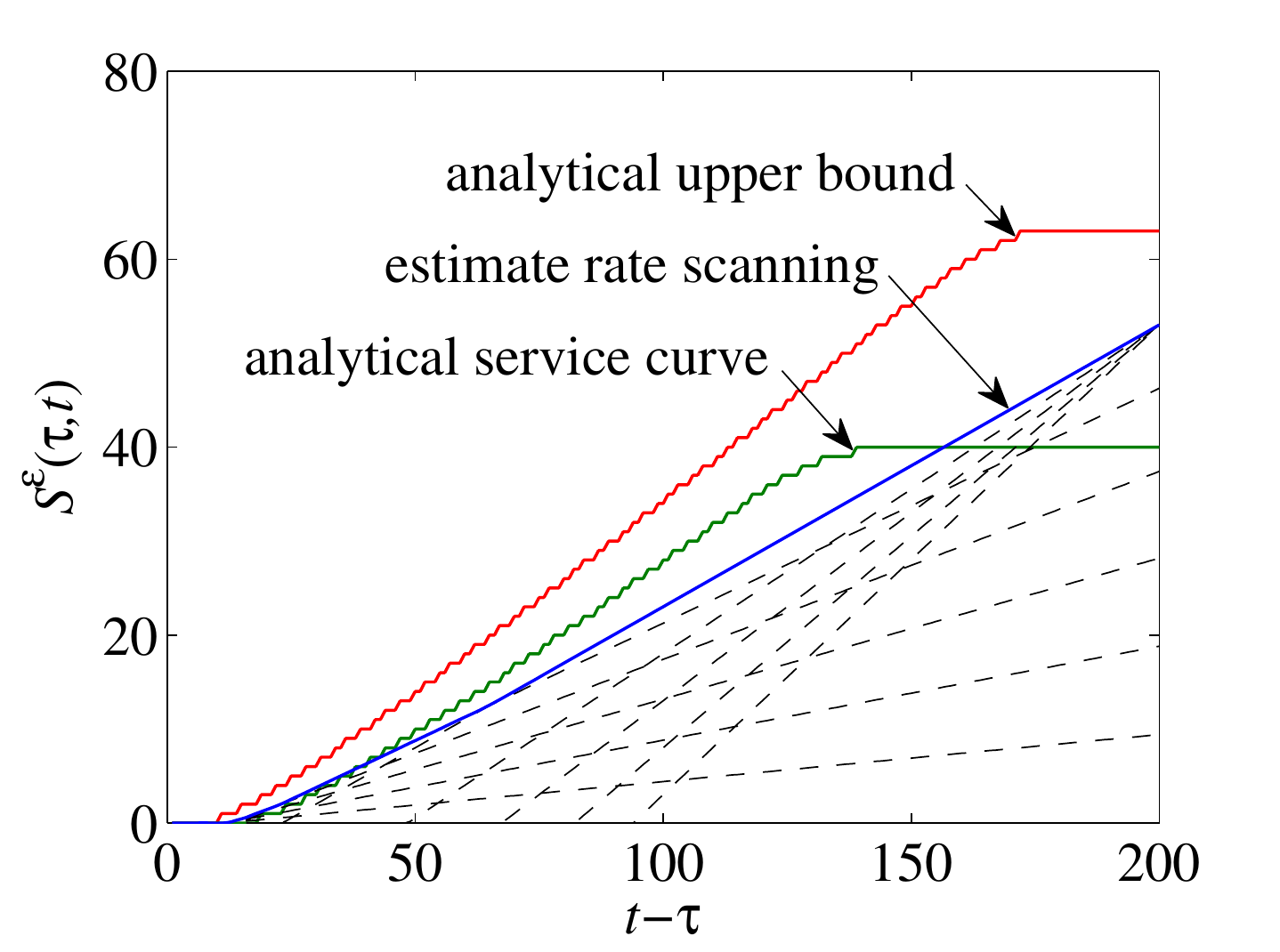}
\label{fig:ratescanning}
}
\hspace{-10pt}
\subfigure[Burst response]{
\includegraphics[width=0.51\columnwidth]{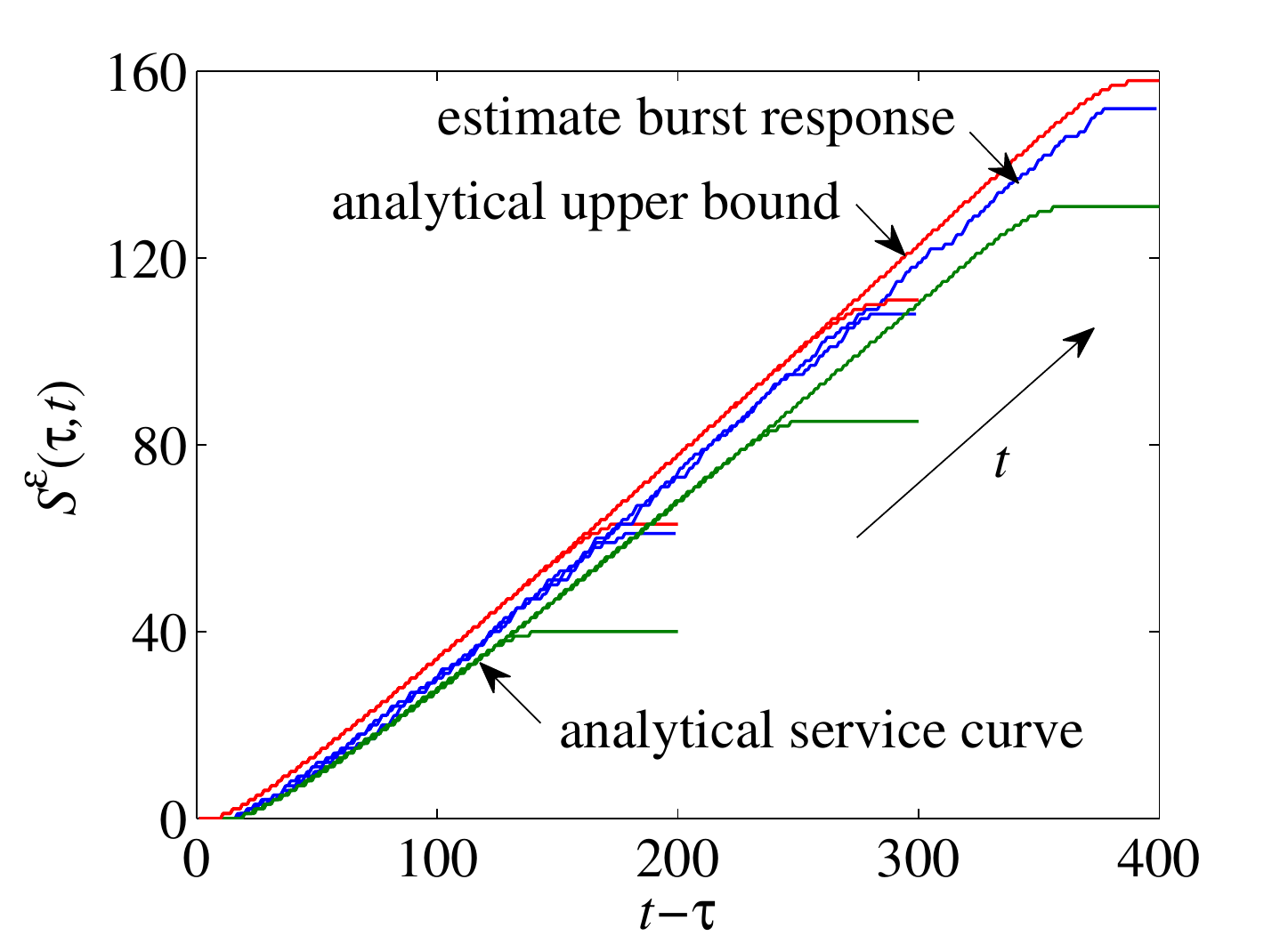}
\label{fig:burstprobing}
\hspace{-10pt}
}
\vspace{-11pt}
\caption{Service curve estimates compared to analytical results. The estimate from rate scanning is the maximum of linear rate segments (dashed lines). By construction it can only recover a convex hull. Burst probing can estimate non-convex service curves and performs close to the analytical upper bound.}
\label{fig:servicecurveestimates}
\vspace{-10pt}
\end{figure}
Fig.~\ref{fig:burstprobing} shows service curve estimates from burst probing for $t=200$, $300$, and $400$ together with analytical upper bounds and analytical reference service curves for the same system as used for Fig.~\ref{fig:ratescanning}. We observe that burst probing performs very well as it provides service curve estimates that closely follow the analytical upper bound.

Despite its good performance, burst probing has, however, limitations. Particularly, we discover that the intuitive assumption of additive service processes~\cite{jiang:stochasticnetworkcalculus}, that is a basis of the method, is not justifiable in general. A notable exception are networks of systems, where the end-to-end network service process $S^{net}(\tau,t) = S^1 \otimes S^2 \otimes \dots \otimes S^n(\tau,t)$ is derived by min-plus convolution of the service processes of the individual systems $S^i(\tau,t)$ for $i = 1,2,\dots n$. For additive $S^i$, the following Lemma~\ref{lem:additivity} proves that $S^{net}$ is super-additive in general and additive only in certain special cases. The result is significant since it refutes the assumption of additive service processes for the large class of tandem systems. It implies that the burst probing method is too optimistic as it estimates $S^{net} (\tau,t) = S^{net}(0,t) - S^{net}(0,\tau)$ from~\eqref{eq:additivity} whereas in general only $S^{net}(\tau,t) \le S^{net}(0,t)-S^{net}(0,\tau)$ holds due to the super-additivity of $S^{net}$. Besides, burst probes can be considered intrusive as they cause non-linear behavior of certain systems~\cite{liebeherr:availbw}. An important example are FIFO multiplexers, where a burst probe can preempt other traffic, resulting in a too optimistic service estimate.
%
%
\subsection{Super-additivity of $\min$ and $\otimes$}
\label{sec:superadditivity}
The burst response method in Sec.~\ref{sec:burstresponse} is applicable to additive service processes. The following lemmas show that the $\min$ and $\otimes$ operators in general maintain only super-additivity but not additivity. A function $f(s,t)$ is super-additive if $f(s,u) \ge f(s,t) + f(t,u)$ for all $u \ge t \ge s \ge 0$. Trivially, an additive function $f(s,u) = f(s,t) + f(t,u)$ for all $u \ge t \ge s \ge 0$ is also super-additive.
\begin{lem}[Super-additivity of $\min$]
\label{lem:additivitymin}
Given two super-additive bivariate functions $f(s,t)$ and $g(s,t)$ for $t \ge s \ge 0$. The minimum $h(s,t) = \min \{f(s,t), g(s,t)\}$ is super-additive.
\end{lem}
We note that the minimum of two additive bivariate functions $f(s,t)$ and $g(s,t)$ for $t \ge s \ge 0$ is super-additive, as a special case of Lemma~\ref{lem:additivitymin}, but in general not additive. A counterexample is $f(s,t) = t-s$ and $g(s,t) = 2(\lfloor t/2 \rfloor - \lfloor s/2 \rfloor)$. Clearly $f$ and $g$ are additive, however, $h = \min \{f,g\}$ is not.
\begin{proof}[Proof of Lemma~\ref{lem:additivitymin}]
By definition of $h$ we have
\begin{align*}
&h(s,t) + h(t,u) \\
&= \min \{f(s,t), g(s,t)\} + \min \{f(t,u), g(t,u)\} \\
&\le \min \{f(s,u), g(s,u), f(s,t)+g(t,u), g(s,t)+f(t,u)\} \\
&\le h(s,u).
\end{align*}
In the second line we used the super-additivity of $f$ and $g$ and in the third line $\min\{f(s,u),g(s,u)\}=h(s,u)$ and $\min\{h(s,u),x\} \le h(s,u)$ for any $x$.
\end{proof}
\begin{lem}[Super-additivity of $\otimes$]
\label{lem:additivity}
Given two bivariate functions $f(s,t)$ and $g(s,t)$ for $t \ge s \ge 0$ where $f(t,t), g(t,t) = 0$ for all $t \ge 0$. Define $h(s,t) = f \otimes g(s,t)$.
\begin{enumerate}
\renewcommand{\theenumi}{\roman{enumi}}
\item If $f$ and $g$ are super-additive, then $h$ is super-additive.
\item If $f$ and $g$ are additive and univariate, then $h$ is additive.
\end{enumerate}
\end{lem}
We note that the convolution of two additive bivariate functions $f(s,t)$ and $g(s,t)$ for $t \ge s \ge 0$ is super-additive, but in general not additive. As a counterexample consider the additive functions $f(s,t) = t-s$ and $g(t,u) = 2(\lfloor u/2 \rfloor - \lfloor t/2 \rfloor)$ for which $h(s,u) = f \otimes g(s,u) = 2 \lfloor u/2 \rfloor - s$ is not additive. To quantify this effect, we define the maximal deviation of a super-additive function $h(s,u)$ from additivity as
\begin{equation}
\Delta(s,u) := h(s,u) - \inf_{t \in [s,u]} \{h(s,t) + h(t,u)\} .
\label{eq:deviationfromadditivity}
\end{equation}
Additivity is, however, achieved in case of univariate functions, where Lemma~\ref{lem:additivity} ii) extends a result for min-plus convolution of sub-additive univariate functions~\cite[p. 142]{leboudec:networkcalculus}.

In Fig.~\ref{fig:superadditivity}, we evaluate the network service process $S^{net}(\tau,t) = S^1 \otimes \dots \otimes S^n(\tau,t)$ for a tandem of $n=2,3,4$ systems with random sleep scheduling as in~Sec.~\ref{sec:systemmodelrandom}, i.e., the service processes $S^i(\tau,t)$ are additive. We consider $10^5$ sample paths and show the distribution of the relative deviation of $S^{net}(0,400)$ from additivity, i.e., $\Delta(\tau,t)/S^{net}(\tau,t)$. Fig.~\ref{fig:superadditivity} confirms significant deviations from additivity.
\begin{proof}[Proof of Lemma~\ref{lem:additivity}]
By definition of $h$ we have
\begin{align}
&h(s,t) + h(t,u) = f \otimes g (s,t) + f \otimes g (t,u) \nonumber \\
&= \inf_{\tau \in [s,t]} \inf_{\upsilon \in [t,u]} \{ f(s,\tau) + f(t,\upsilon) + g(\tau,t) + g(\upsilon,u) \}.
\label{eq:doubleinfadditivity}
\end{align}

i) Given $f$ and $g$ are super-additive. From~\eqref{eq:doubleinfadditivity} we have
\begin{align}
&h(s,t) + h(t,u) \nonumber \\
&\le \inf_{\tau \in [s,t]} \inf_{\upsilon \in [t,u]} \{ f(s,\upsilon) - f(\tau,t) + g(\tau,t) + g(\upsilon,u) \} \nonumber \\
&= \inf_{\upsilon \in [t,u]} \{ f(s,\upsilon) + g(\upsilon,u) \} + \inf_{\tau \in [s,t]} \{ g(\tau,t) - f(\tau,t) \} \nonumber \\
&\le \inf_{\upsilon \in [t,u]} \{ f(s,\upsilon) + g(\upsilon,u) \}.
\label{eq:subadditivitypart1}
\end{align}
In the first line, we estimated $f(s,\tau) + f(\tau,t) + f(t,\upsilon) \le f(s,\upsilon)$ due to the super-additivity of $f$. In the second line, we rearranged the infima, and in the third line, we estimated $\inf_{\tau \in [s,t]} \{ g(\tau,t) - f(\tau,t) \} \le g(t,t) - f(t,t) = 0$ since $f(t,t),g(t,t)=0$ for all $t \ge 0$. Similarly, using the super-additivity of $g$ we derive from~\eqref{eq:doubleinfadditivity} that
\begin{equation}
h(s,t) + h(t,u) \le \inf_{\tau \in [s,t]} \{ f(s,\tau) + g(\tau,u) \} .
\label{eq:subadditivitypart2}
\end{equation}
Combining~\eqref{eq:subadditivitypart1} and~\eqref{eq:subadditivitypart2} we obtain
\begin{equation*}
h(s,t) + h(t,u) \le \inf_{\tau \in [s,u]} \{ f(s,\tau) + g(\tau,u) \} = h(s,u) ,
\end{equation*}
which proves the super-additivity of $h$.

ii) For the special case of univariate functions $f(s,t) = f(t-s)$ and $g(s,t) = g(t-s)$, that depend only on the difference $t-s$ and not on the absolute values of $s$ and $t$, it follows that $h(s,t) = f \otimes g(t-s) = h(t-s)$ is also univariate. Using the additivity of $f$ and $g$,~\eqref{eq:doubleinfadditivity} yields that
\begin{align*}
& h(t-s) + h(u-t) \\
&= \inf_{\tau \in [s,t]} \inf_{\upsilon \in [t,u]} \{ f(\tau-s+\upsilon-t) + g(t-\tau+u-\upsilon) \} \\
&= \inf_{\varsigma \in [s+t,t+u]} \{ f(\varsigma-s-t) + g(t+u-\varsigma) \} \\
&= \inf_{\varsigma \in [0,u-s]} \{ f(\varsigma) + g(u-s-\varsigma)\} = h(u-s)
\end{align*}
which proves the additivity of $h$.
\end{proof}
\begin{figure}
\hspace{-10pt}
\subfigure[Deviation from additivity]{
\includegraphics[width=0.51\columnwidth]{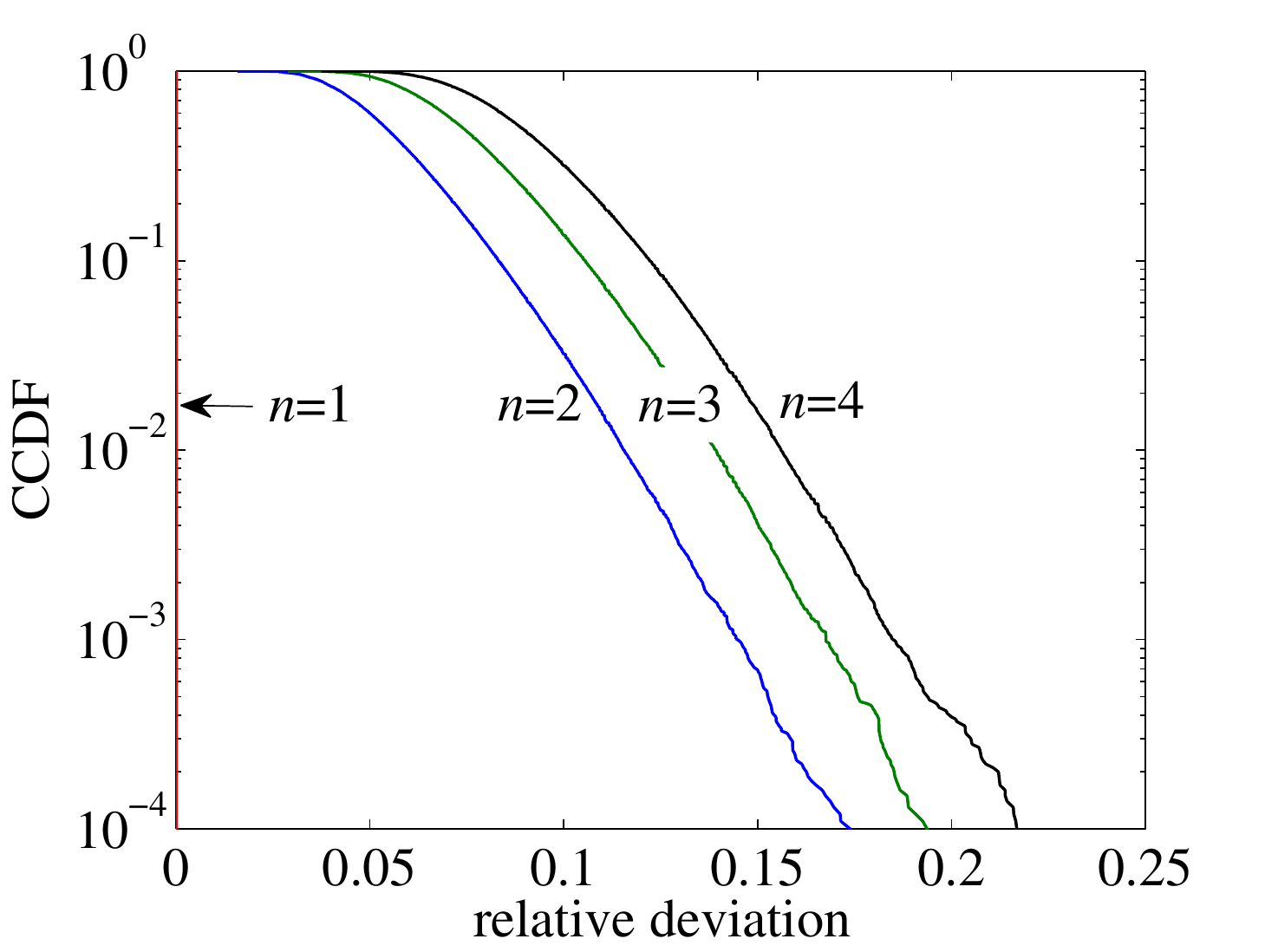}
\label{fig:superadditivity}
}
\hspace{-10pt}
\subfigure[Accuracy of minimal probing]{
\includegraphics[width=0.51\columnwidth]{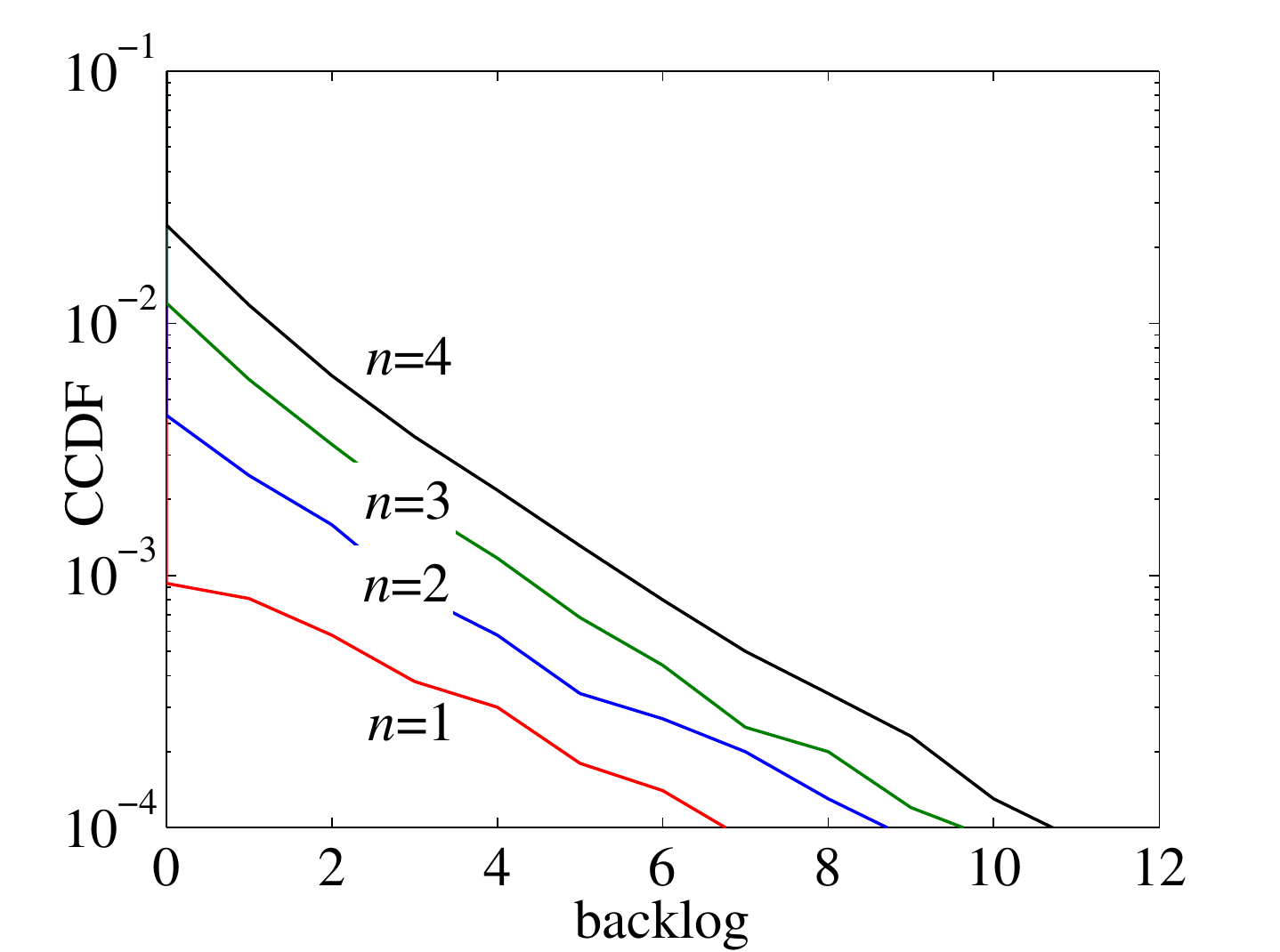}
\label{fig:minimalprobing}
\hspace{-10pt}
}
\vspace{-10pt}
\caption{Network of $n$ systems with random sleep scheduling in series. (a) The network service process deviates from additivity. (b) Minimal probing achieves small backlogs, corresponding to a high accuracy of the estimate.}
\label{fig:deviationsfromadditivity}
\vspace{-11pt}
\end{figure}
%
%
\subsection{Minimal Probing}
\label{sec:minimalprobing}
Next, we devise a new probing method to characterize the transient behavior of service curves that are neither convex nor additive. The method comprises two phases. In the first step, the method estimates a minimal probe as well as an upper bound of the service curve from burst probing. In the second step, the minimal probe is used to estimate a non-stationary service curve with a defined accuracy. We show that the minimal probe can reveal the service of the system, ideally without causing backlogs. The importance of the probe traffic intensity is also discussed in~\cite{liebeherr:availbw} and estimation methods for stationary systems that use minimal backlogging are~\cite{valaee:adhocadmissioncontrol, nam:minimalbackloggingbwest}.
\subsubsection{Estimation using Arbitrary Probes}
First, we consider how to obtain a service curve estimate from an arbitrary probe. The estimate will then be used to derive conditions for the shape of a minimal probe. From~\eqref{eq:linearsystem} it follows that $D(t) \le A(\tau) + S(\tau,t)$ for all $\tau \in [0,t]$ so that
\begin{equation}
S(\tau,t) \ge D(t) - A(\tau),
\label{eq:estimategeneral}
\vspace{-1pt}
\end{equation}
for all $\tau \in [0,t]$. An equivalent expression using the backlog is $S(\tau,t) \ge A(\tau,t)-B(t)$ for all $\tau \in [0,t]$. By insertion of the backlog quantile it follows that
\begin{equation}
\mathcal{S}^{\varepsilon}(\tau,t) = A(\tau,t) - B^{\varepsilon}(t)
\label{eq:estimateminimalprobing}
\vspace{-1pt}
\end{equation}
satisfies~\eqref{eq:bivariateenvelope}, i.e., $\mathcal{S}^{\varepsilon}(\tau,t)$ is a non-stationary service curve~\eqref{eq:servicecurve}. For the special case of $A(\tau,t) = r (t-\tau)$ the rate scanning method is recovered.
\subsubsection{Definition of Minimal Probes}
So far, we did not constrain the shape of the probes and defining a suitable probe is non-trivial. Intuitively, a too small probe will provide little information about the service as the observed departures are mostly limited by the arrivals~\cite{liebeherr:availbw}. A too large probe, on the other hand, will deteriorate the estimate, e.g., in the extreme case of a burst probe $A(\tau) = \delta(\tau)$, the lower bound in~\eqref{eq:estimategeneral} will only be useful for $\tau=0$ but not for $\tau > 0$. The same restriction applied to burst probing earlier. Consequently, we seek to find the minimal probe $A_{mp}(\tau)$ that satisfies $D(t) = \inf_{\tau \in [0,t]} \{A_{mp}(\tau) + S(\tau,t) \} = S(0,t)$ for a fixed $t > 0$, i.e., the minimal probe that allows estimating the service from observations of the departures. The minimal probe is
\begin{equation}
A_{mp}(\tau) = S(0,t) - S(\tau,t)
\label{eq:minimalprobe}
\vspace{-1pt}
\end{equation}
since any function $A(\tau)$ for which $\exists \tau \in [0,t]: A(\tau) < A_{mp}(\tau)$ will only reveal $D(t) < S(0,t)$. We note that $A_{mp}(\tau)$ maximizes the right hand side of~\eqref{eq:estimategeneral}, where~\eqref{eq:linearsystem} confirms that $D(t)-A_{mp}(\tau) = S(\tau,t)$. For any other larger or smaller probe $A(\tau) = A_{mp}(\tau) + f(\tau)$ it holds that $D(t)-A(\tau) = S(\tau,t) + \inf_{\upsilon \in [0,t]} \{f(\upsilon) \} - f(\tau) \le S(\tau,t)$.

Since $A_{mp}(\tau)$ depends on the unknown service, it cannot be constructed a priori. To gather information that facilitates estimating $A_{mp}(\tau)$, we initially consider the system's burst response $D(t) = S(0,t)$ for $A(t) = \delta(t)$. An estimate of~\eqref{eq:minimalprobe} is $\widetilde{A}_{mp}(\tau) = S(0,\tau)$ \mbox{that is exact if $S(\tau,t)$ is} additive. In the stochastic case, we use $\mathcal{S}_{br}^{\varepsilon}(\tau,t)$ from~\eqref{eq:effectiveserviceburstrespones} to estimate
\begin{equation}
\widetilde{A}_{mp}(\tau) = \mathcal{S}_{br}^{\varepsilon}(0,t) - \mathcal{S}_{br}^{\varepsilon}(\tau,t).
\label{eq:minimalprobestochastic}
\end{equation}
We note that additivity of $\mathcal{S}_{br}^{\varepsilon}(\tau,t)$ cannot be assumed by construction of~\eqref{eq:effectiveserviceburstrespones}, see Lemma~\ref{lem:additivitymin}.
\subsubsection{Accuracy of the Estimates}
We start the investigation for a time-variant but otherwise deterministic system, where~\eqref{eq:estimategeneral} provides a lower estimate of the service. Given the probe $\widetilde{A}_{mp}(\tau) = S(0,\tau)$, where $S(0,\tau)$ is obtained as the system's burst response. By insertion into~\eqref{eq:estimategeneral}, a lower estimate is
\begin{equation*}
S(\tau,t) \ge \inf_{\upsilon \in [0,t]} \{S(0,\upsilon) + S(\upsilon,t)\} - S(0,\tau),
\end{equation*}
where we used~\eqref{eq:linearsystem} to compute $D(t)$ for $\widetilde{A}_{mp}(\tau)$. On the other hand, the burst response provides the upper estimate $S(\tau,t) \le S(0,t) - S(0,\tau)$ if $S(\tau,t)$ is super-additive. Taking the difference of the upper and the lower estimate, the service process is bounded in an interval of width
\begin{equation*}
\Delta(t) = S(0,t) -  \inf_{\upsilon \in [0,t]} \{S(0,\upsilon) + S(\upsilon,t)\}
\end{equation*}
that is the maximum deviation of $S(0,t)$ from additivity as defined in~\eqref{eq:deviationfromadditivity}. Further, using $\widetilde{A}_{mp}(t) = S(0,t)$ and~\eqref{eq:linearsystem},
\begin{equation*}
\Delta(t) = \widetilde{A}_{mp}(t) - D(t) = B(t),
\end{equation*}
i.e., the estimation accuracy can be directly seen from the backlog at the end of the probe. Consequently, if $S(\tau,t)$ is additive, $\Delta(t)=0$, i.e., both the lower and the upper estimate recover $S(\tau,t)$ exactly and the backlog $B(\tau)$ that is caused by $\widetilde{A}_{mp}(\tau)$ is zero during the entire probe, i.e., for all $\tau \in [0,t]$.

In the stochastic case, we obtain a non-stationary service curve estimate from~\eqref{eq:estimateminimalprobing} as $\mathcal{S}_{mp}^{\varepsilon}(\tau,t) = \widetilde{A}_{mp}(\tau,t) - B^{\varepsilon}(t)$ using the probe defined in~\eqref{eq:minimalprobestochastic}. By use of~\eqref{eq:minimalprobestochastic} we have $\widetilde{A}_{mp}(\tau,t) = \widetilde{A}_{mp}(t) - \widetilde{A}_{mp}(\tau) = \mathcal{S}_{br}^{\varepsilon}(\tau,t)$ since $\mathcal{S}_{br}^{\varepsilon}(t,t) = 0$ by definition. Note that we do not assume additivity of $S(\tau,t)$, so far. By insertion into~\eqref{eq:estimateminimalprobing}, it holds that
\begin{equation}
\mathcal{S}_{mp}^{\varepsilon}(\tau,t) = \mathcal{S}_{br}^{\varepsilon}(\tau,t) - B^{\varepsilon}(t).
\label{eq:estimateaccuracy}
\end{equation}
We conclude that $B^{\varepsilon}(t)$ observed by minimal probing is a measure of accuracy that separates the conservative estimate of minimal probing from the possibly too optimistic estimate of burst probing.

To investigate $B^{\varepsilon}(t)$, consider the backlog expression $B(t) = \sup_{\tau \in [0,t]} \{ A(\tau,t) - S(\tau,t) \}$ that is derived from~\eqref{eq:linearsystem}. By insertion of~\eqref{eq:minimalprobestochastic} we have $B(t) = \sup_{\tau \in [0,t]} \{ \mathcal{S}_{br}^{\varepsilon}(\tau,t) - S(\tau,t) \}$. By definition~\eqref{eq:effectiveserviceburstrespones}, $\mathcal{S}_{br}^{\varepsilon}(\tau,t) = \inf_{\psi \in \Psi_t} \{S_{\psi}(0,t) - S_{\psi}(0,\tau)\}$ so that for any sample path $\varphi \in \Psi_t$
\begin{align*}
B_{\varphi}(t) &= \sup_{\tau \in [0,t]} \Bigl\{ \inf_{\psi \in \Psi_t} \{S_{\psi}(0,t) - S_{\psi}(0,\tau)\} - S_{\varphi}(\tau,t) \Bigr\} \\
&\le \sup_{\tau \in [0,t]} \{ S_{\varphi}(0,t) - S_{\varphi}(0,\tau) - S_{\varphi}(\tau,t) \} \\
&= S_{\varphi}(0,t) - \inf_{\tau \in [0,t]} \{S_{\varphi}(0,\tau) + S_{\varphi}(\tau,t) \},
\end{align*}
i.e., $B_{\varphi}(t)$ is bounded by the maximal deviation of $S_{\varphi}(0,t)$ from additivity. Finally, if $S(\tau,t)$ is additive, it follows that $B_{\varphi}(t) = 0$ for all $\varphi \in \Psi_t$. Since $\mathsf{P}[\Psi_t] \ge 1-\varepsilon$ it holds that $B^{\varepsilon}(t) = 0$ and $\mathcal{S}_{mp}^{\varepsilon}(\tau,t)$ recovers $\mathcal{S}_{br}^{\varepsilon}(\tau,t)$ exactly.

Fig.~\ref{fig:minimalprobing} quantifies the distribution of $B(t)$ that is observed by minimal probing at $t=400$. A network of $n = 1 \dots 4$ systems in series, each with random sleep scheduling as in Sec.~\ref{sec:systemmodelrandom}, is considered. The network service process is additive for $n=1$, but not for $n > 1$ as confirmed by Fig.~\ref{fig:superadditivity}. Since in general, it is not known whether the service process is additive or not, estimates of burst probing are not reliable. Using minimal probing, we can either if $B^{\varepsilon}(t) = 0$ confirm the estimate of burst probing or otherwise obtain a conservative estimate with a defined accuracy that is given by $B^{\varepsilon}(t)$. In Fig.~\ref{fig:minimalprobing} we observe that $B^{\varepsilon}(t) > 0$ for $\varepsilon = 10^{-3}$ and $n > 1$, i.e., the estimate from burst probing is not confirmed. $B^{\varepsilon}(t)$ is, however, small, i.e, minimal probing is accurate. We omit showing further service curve estimates of minimal probing since, except for the additive offset $B^{\varepsilon}(t)$, they generally match the estimates of burst probing as determined by~\eqref{eq:estimateaccuracy}.
%
%
\section{Conclusions}
\label{sec:conclusion}
This work contributed a notion of non-stationary service curves to the stochastic network calculus. We analyzed systems with sleep scheduling and provided insights into the transient behavior. Noticeable transient overshoots and large relaxation times were observed that have a significant impact on network performance. Due to its generality, the service curve model can include a variety of further systems. Beyond modelling, we considered measurement-based methods for identification of systems. We discovered that existing probing methods cannot accurately estimate non-stationary service curves. The difficulties were related to the non-convexity and super-additivity of the service. We devised a novel two-phase method that first estimates a minimal probe that is used in a second step to obtain an accurate service curve estimate. Simulation results confirmed the fidelity of the approach.
%
%
\bibliographystyle{IEEEtran}
\bibliography{IEEEabrv,IEEEfidler}
\end{document}